\documentclass[onecolumn, techreport]{IEEEtran}
\usepackage[latin9]{inputenc}
\usepackage[active]{srcltx}
\usepackage{amsmath}
\usepackage{amssymb}
\usepackage{graphicx}

\makeatletter
\usepackage[top=1.2in, bottom=1.2in, left=1.5in, right=1.5in]{geometry}

\IEEEoverridecommandlockouts

\usepackage{multicol}\usepackage{subfigure}\usepackage{cases}\usepackage{color}\usepackage{epsfig}\usepackage{algorithm}\usepackage{algorithmic}%\usepackage{hyperref}
\usepackage{url}

\usepackage{cite}% cite.sty was written by Donald Arseneau
\ifCLASSINFOpdf
\else
\fi
\makeatother

\begin{document}

\title{Real-Time Peer-to-Peer Streaming Over Multiple Random Hamiltonian
Cycles}
%\author{Joohwan Kim~\IEEEmembership{Member,~IEEE} and R. Srikant,~\IEEEmembership{Fellow,~IEEE}\thanks{Joohwan Kim and R. Srikant are with the Department of Electrical and Computer Engineering and Coordinated Science Laboratory at the University of Illinois at Urbana-Champaign. This paper was presented in part at the Information Theory and Applications Workshop \cite{Kim12ITA}. This work has been supported by the National Science Foundation grant CNS 09-64081 and Army MURIs W911NF-08-1-0233 and W911NF-07-1-0287.}}
\author{\IEEEauthorblockN{Joohwan Kim\IEEEauthorrefmark{1} and R. Srikant\IEEEauthorrefmark{2}}\IEEEauthorblockA{\\
Dept. of Electrical and Computer Engineering\\
 and Coordinated Science Laboratory\\
 University of Illinois at Urbana-Champaign\\
 Email: \{\IEEEauthorrefmark{1}joohwan, \IEEEauthorrefmark{2}rsrikant\}@illinois.edu}
\thanks{This work has been supported by the National Science Foundation grant CNS 09-64081 and Army MURIs W911NF-08-1-0233 and W911NF-07-1-0287.
}}
\maketitle

\begin{abstract}
We are motivated by the problem of designing a simple distributed algorithm for Peer-to-Peer streaming applications that can achieve high throughput and low delay, while allowing the neighbor set maintained by each peer to be small. While previous works have mostly used tree structures, our algorithm constructs multiple random directed Hamiltonian cycles and disseminates content over the superposed graph of the cycles. We show that it is possible to achieve the maximum streaming capacity even when each peer only transmits to and receives from $\Theta(1)$ neighbors. Further, we show that the proposed algorithm achieves the streaming delay of $\Theta(\log N)$ when the streaming rate is less than $(1-1/K)$ of the maximum capacity for any fixed constant $K\geq2$, where $N$ denotes the number of peers in the network. The key theoretical contribution is to characterize the distance between peers in a graph formed by the superposition of directed random Hamiltonian cycles, in which edges from one of the cycles may be dropped at random. We use Doob martingales and graph expansion ideas to characterize this distance as a function of $N$, with high probability.
\end{abstract}


% \fi
\IEEEpeerreviewmaketitle

%%%  Alleviated Math Term   %%%%%%%
%\newcommand{\cN}{\mathcal{N}}
%\newcommand{\cL}{\mathcal{L}}
\global\long\def\cS{\mathcal{S}}
 \global\long\def\cC{\mathcal{C}}
\newcommand{\cG}{\mathcal{G}}

\newcommand{\sfactor}{F}
\newcommand{\sfactorl}[1]{\tilde F^{(#1)}}
\newcommand{\NN}{\nonumber}

\newcommand{\summu}[1]{\tblue{ \sum_{j=1}^{#1} \tau_j}}
\newcommand{\shrink}{F}
 \global\long\def\ex{\textrm{exp}}
 \global\long\def\comb#1#2{ { { #1 }  \choose {#2} }}
 \global\long\def\hof#1{H\big( #1 \big)}
 \global\long\def\hoff#1#2{\hof{\frac{#1}{#2}}}
 \global\long\def\cut{\textrm{cut}}
 \global\long\def\indicate#1{1_{\{#1\}}}
 \global\long\def\depth{d}
 \global\long\def\distance{\text{dist}}
 \global\long\def\prob#1{P\left[#1 \right]}
 \global\long\def\Gmm{G_{M,m}}
 \global\long\def\layer{L}
 \global\long\def\layermk{\mathcal{L}_{M,k}}

\newcommand{\tblue}{\textcolor{black}}
\newcommand{\tred}[1]{}

%Activate if you see both journal and tech sentences
%with distinctive colors
%\newcommand{\tech}[1]{\textcolor{green}{($<$TECH. VER. STARTS)}#1\textcolor{green}{(TECH. VER. ENDS.$>$)}}
%\newcommand{\tech}[1]{{($<$TECH. VER. STARTS)}#1{(TECH. VER. ENDS.$>$)}}
%\newcommand{\tech}[1]{}

%\newcommand{\jour}{\textcolor{magenta}}
%\newcommand{\jour}[1]{\textcolor{magenta}{($<$JOUR. VER. STARTS)}#1\textcolor{magenta}{(JOUR. VER. ENDS.$>$)}}
%\newcommand{\jour}[1]{{($<$JOUR. VER. STARTS)}#1{(JOUR. VER. ENDS.$>$)}}
%\newcommand{\jour}{\textcolor{black} }

\newtheorem{proposition}{\textbf{Proposition}} \newtheorem{assumption}{\textbf{Assumption}}
\newtheorem{lemma}{\textbf{Lemma}}
\newtheorem{corollary}{\textbf{Corollary}} \newtheorem{theorem}{\textbf{Theorem}}
\newtheorem{condition}{\textbf{Condition}}
\newlength{\mylength}
\setlength{\mylength}{5in}

\section{Introduction}

%In this paper, we consider a Peer-to-Peer (P2P) streaming network with one source peer broadcasting real-time content and other peers watching the content. The source peer sequentially generates \emph{chunks} from the real-time content, and sends the chunks to its neighboring peers. Other peers then exchange the chunks among themselves, consuming their own upload bandwidth. Due to this bandwidth contribution of the ordinary peers, it is theoretically possible in a P2P network that the source broadcasts the content to a large number of peers using a small upload bandwidth which does not increase with the network size. The questions that naturally arise are ``How to construct a network topology for P2P streaming?'' and ``How to disseminate the chunks over the topology?'' Solutions to these questions critically determine throughput and delay of chunk dissemination, which affect the streaming quality experienced by each peer.

Dissemination of multimedia content over the Internet is often accomplished using a central server or a collection of servers which disseminate the data to all clients interested in the content. Youtube is an example of such a model, where multiple large-capacity servers are used to meet the download demands of millions of users. In contrast, in a peer-to-peer (P2P) network, a small (low-capacity) server uploads the content to a small number of clients, and these clients and all other clients in the network then exchange content among themselves.
The P2P approach is scalable since the network utilizes the upload capacities of all the clients (commonly known as peers) in the network: as more peers join the network, the download requirement increases but the available capacity also increases proportionally. In this paper, we are interested in designing P2P networks where each peer needs to keep track of only a small number of other peers in the network. Restricting the neighborhood size of each peer reduces the administrative overhead for the peers. Therefore, one of the key challenges is to design an algorithm to decide which peers should
belong to the same neighborhood. Such algorithms are called ``pairing''
algorithms since they pair peers to be neighbors. The pairing algorithm
must be lightweight, i.e., when new peers enter the network or when
existing peers exit the network, the algorithm should incur low overhead
to readjust the pairing relationships in the neighborhood.
In this paper, we propose a pairing algorithm based on directed Hamiltonian cycles which has low overhead for node insertion and deletion. The insertion and deletion parts of our algorithm are the same as the algorithm proposed in \cite{Law03} for constructing small diameter graphs using undirected Hamiltonian cycles for distributed hash table (DHT) applications. However, there are certain key differences: our algorithm requires edges to be directed for real-time streaming purposes and while small diameter is sufficient for fast lookup times in DHTs, it is not sufficient to ensure high throughput and low delay for streaming applications.

The pairing algorithm mentioned in the previous paragraph determines
the topology of the network. Given the topology, the network must
then decide how to disseminate content in the network to achieve the
maximum possible capacity and low delay. Multimedia content is often
divided into chunks and thus, the content dissemination algorithm
is also called the chunk dissemination algorithm in the literature.
Chunk dissemination is accomplished by a peer in two steps in each
timeslot: the peer has to select a neighbor to receive a chunk (called
\emph{neighbor selection}) and then it has to decide which chunk it
will transmit to the selected neighbor (called \emph{chunk selection}).
Thus, the practical contributions of the paper can be summarized as
follows: we present a low-complexity, high-throughput and low-delay
algorithm for pairing, neighbor selection and chunk selection in real-time
P2P streaming networks. We emphasize that the goal of this paper is
to study real-time data dissemination in P2P networks. This is in
contrast to stored multimedia content dissemination (which is the
bulk of Youtube's data, for example) or file-transfer applications
(such as in BitTorrent).

Our approach for pairing results in a graph formed by the superposition
of multiple random directed Hamiltonian cycles over a given collection
of nodes (peers). We will see that the performance analysis of our
algorithms requires us to understand the distance (the minimum number
of hops) from a given peer to all other peers in the graph. The main
theoretical contribution of the paper is to characterize these distances
with high probability through a concentration result using Doob martingales.
Using this result, we show that our algorithm achieves $\Theta(\log N)$
delay with high probability, when the streaming rate is less than or  equal to $(1-\frac{1}{K})$ of the optimal capacity for any constant
$K\geq2$, where $N$ denotes the number of peers in the network.

This paper is organized as follows. In Section~\ref{sec:related_work},
we review prior work in the area of real-time P2P networks. In Section~\ref{sec:examples},
we provide two examples to help the reader understand the advantage
of using random Hamiltonian cycles. In Section~\ref{sec:System_Model},
we present our P2P algorithm  that constructs random Hamiltonian
cycles and disseminates content over the cycles in a fully distributed manner. In Section~\ref{sec:Througput_Analysis},
we consider the streaming rate that can be achieved under our algorithm.
In Sections~\ref{sec:Delay_Analysis} and~\ref{sec:diameter}, we analyze the delay to disseminate
chunks to all peers under our algorithm.
 In Section~\ref{sec:Conclusions},
we conclude the paper.

\section{Related Work}

\label{sec:related_work}

We briefly review prior work in the area of real-time P2P networks.
Prior work in the area can be broadly categorized as designing one
of two types of networks: a structured P2P network or an unstructured
P2P network. The structured P2P streaming approach focuses on constructing
multiple overlay spanning trees that are rooted at the source \cite{Chu65,Edmonds1972,Tarjan77,Castro2003,Li05,Mundinger08 ,Kumar07,Liu08,Liu10}.
In this approach, the real-time content arriving at the source is
divided into multiple sub-streams and each sub-stream is delivered
over one of the trees. Since this approach uses the tree structure,
connectivity from the source to all peers is guaranteed. By managing
the tree depth to be $\Theta(\log N)$, this approach can guarantee
$\Theta(\log N)$ delay to disseminate a chunk of each sub-stream
to all peers. However, the fundamental limitation of the structured
P2P streaming is vulnerability to peer churn. It is well known that
the complexity of constructing and maintaining $\Theta(\log N)$-depth
trees grows as $N$ increases \cite{Liu08,Liu10}. Therefore, in a highly dynamic P2P
network where peers frequently join and leave the network, the structure
approach is not scalable.

%For an edge-capacitated network, the optimal streaming capacity is known to  be minimum cut capacity of the graph. Several algorithms that construct spanning trees from the given edges are proposed to  achieve this optimal capacity For a node-capacitated network (which models the upload and download bandwidth constrained Internet better than the edge-capacitated network), the optimal streaming capacity is known to be the average upload capacity of peers. There are also several algorithms that construct throughput optimal trees in the node capacitated model . However, the trees used in these works are impractical because the degree of the trees are not bounded.
% In \cite{Liu08,Liu10}, the authors proposed the Snowball and Bubble algorithms, respectively, that constructs throughput optimal spanning trees that have a bonded degree and $\Theta(\log N)$ depth. However, the fundamental limitation of all aforementioned tree constructing algorithms is that the complexity of constructing and maintaining trees grows as $N$ increases. Specially in a highly dynamic P2P network where peers frequently join and leave the network, this limitation plays as a critical obstacle to scale the network.

Unstructured P2P networks overcome this vulnerability to peer churn.
In unstructured P2P networks, peers find their neighboring peers randomly
and get paired with them locally. As a neighboring peer leaves, a
peer chooses another peer randomly as its new neighboring peer. Due
to the distributed fashion of this peer pairing, unstructured P2P
networks are robust to peer churn, unlike the structured P2P networks.
However, the fundamental limitation of unstructured P2P networks is
weak connectivity. Since peers are paired randomly without considering
the entire network topology, there may be some peers that are not
strongly connected from the source, which results in poor throughput
and delay. To ensure full connectivity in this approach, it is required
that every peer should be paired with $\Theta(\log N)$ neighboring
peers \cite{Zhao11}, or should constantly change their neighbors
to find neighbors providing a better streaming rate \cite{Zhang05coolstreaming}.
However, in these approaches, delay performance is hard to guarantee
because chunks have to be disseminated over an ``unknown'' network
topology.

Another interesting line of work has studied gossip-based algorithms
that disseminate information to all peers in a fashion similar to
the spread of epidemics. By studying the dissemination delay under
these gossip-based algorithms, we can analyze the delay for peers
to disseminate chunks to all peers in a P2P network. The seminal work
in \cite{Frieze1985} shows that gossiping requires $\Theta(\log N)$
time with high probability to disseminate a single chunk from the
source to all peers. When there is a sequence of chunks arriving at
the source, \emph{the latest-blind algorithm} proposed in \cite{Sanghavi07}
is proven to deliver $(1-e^{-1})$ fraction of chunks to all peers
with $\Theta(\log N)$ delay with high probability.
Later work in \cite{Bonald08} proposed \emph{the latest useful algorithm} that can deliver almost all chunks with $\Theta(\log N)$ delay with high probability. However, the basic assumption for analysis in this line of work is that the network is a complete graph where every peer has $N-1$ outgoing edges to all other peers, and only simulations  are used in \cite{Bonald08} to evaluate the performance on a random graph with bounded degree.
In contrast,  it is shown in \cite{Shah09} that
gossip-based algorithms can achieve $\Theta(\log N)$ delay, when the matrix representing the connectivity between peers is doubly
stochastic and symmetric. However, only a small fraction of the optimal throughput can be guaranteed
with $\Theta(\log N)$ delay.

We address all the aforementioned limitations using multiple
random Hamiltonian cycles. While the structure of Hamiltonian cycles
provides us with full connectivity from the source to all peers, random
pairing within each cycle enables peers to cope with peer churn. Furthermore,
the proposed chunk dissemination algorithm guarantees $\Theta(\log N)$
delay required for each chunk to be disseminated to all the peers
for a near optimal throughput. One may be concerned about using cycles because the diameter $N$ of a cycle could result in poor
delay performance. However, we address this concern  in the next section.

\section{Independent Random Hamiltonian Cycles}

\label{sec:examples}

In a delay-sensitive application, such as P2P streaming, cycles (or
line topologies) have been considered to be undesirable since their
diameter is $N-1$, where the diameter of a directed graph is defined
as the maximum distance between any pair of nodes. Delivering information
from a node to all the other nodes over a cycle requires $N-1$ successive
transmissions, which results in $\Theta(N)$ delay. In this section,
we consider two examples which show that one can use a superposed
graph of multiple cycles as an alternative to the tree structure for
information dissemination.

\begin{figure}[ht]
\centering \subfigure{\label{fig:hamilton1} \includegraphics[width=0.4\mylength]{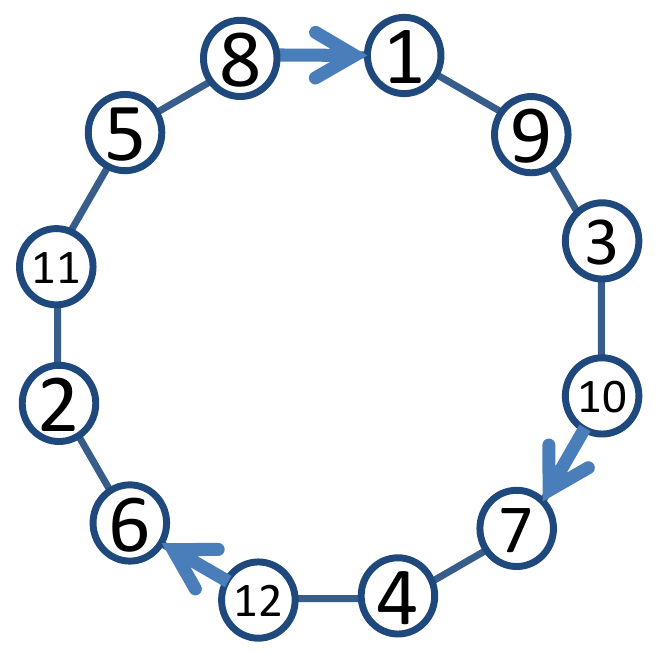}}
\subfigure{\label{fig:hamilton2} \includegraphics[width=0.4\mylength]{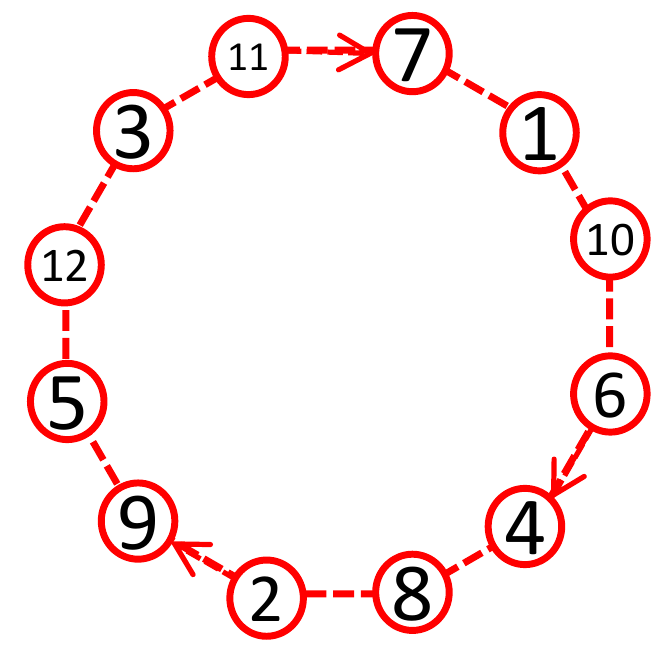}}
\caption{Two random Hamiltonian cycles $(N=12)$ generated by independent permutations
of nodes $1,2,\cdots,N$.}
\end{figure}

Consider $N$ nodes numbered $1,2,\cdots,N$. By permuting these nodes,
we can make a random Hamiltonian directed cycle as shown in Fig.~\ref{fig:hamilton1}.
\emph{(Since all the graphs that we will consider are directed graphs, we
will skip mentioning ``directed'' from now on.)} Make another random
Hamiltonian cycle by independently permuting the $N$ nodes as shown
in Fig.~\ref{fig:hamilton2}. Clearly, the diameter of each cycle
is $N-1$. An interesting question is the following: ``if we superpose
both cycles, what is the diameter of the superposed graph?'' Interestingly,
the diameter significantly reduces from $\Theta(N)$ to $\Theta(\log N)$
with high probability in the superposed graph.%
\footnote{In random graph theory, it is shown that the superposition of two
undirected random Hamiltonian cycles has a distribution similar to
an undirected random regular graph \cite{Kim01}. This regular graph
is known to have $\Theta(\log N)$ diameter with high probability
\cite{Bollobas82}. Combining both results, we can infer that the
superposition of two random undirected Hamiltonian cycles has $\Theta(\log N)$
depth with high probability. From this, it is not very difficult to
obtain a similar result for directed cycles. However, in \cite{Kim11tech2},
we establish this result more directly.%
}

Next, we consider a further modification of the two random cycle model.
From the second cycle (Fig.~\ref{fig:hamilton2}), we remove each
edge with some probability $0<q<1$ independently.  If we superpose the first cycle and the remaining
edges in the second cycle,  what will be the diameter of the graph? Since we have removed around $qN$ edges from
the second cycle, the diameter will certainly increase. However, we will show in a later section
that the order of the diameter still remains $\Theta(\log N)$.

These two examples imply that a graph formed from  superposed Hamiltonian cycles has a small diameter of $\Theta(\log N)$. This means that the superposed graph can be a good alternative to a spanning tree with a bounded outdegree
that has been widely used to achieve a logarithmic dissemination delay in P2P streaming. However, in the case of peer churn, the complexity of constructing and updating spanning trees
(as in prior literature) subject to the constraints on the degree
bound and  the logarithmic depth increases dramatically with the
network size. In contrast, the superposed graph is robust to peer churn because independent cycles are much easier to maintain.
In the rest of this paper, we show how these properties of random superposed cycles
can be used to construct a P2P network that can achieve high throughput
and low delay.

\section{System Model}

\label{sec:System_Model}

We assume that time is slotted, and every peer (including the source
peer) in the network contributes a unit upload bandwidth, i.e., each
peer can upload one chunk per timeslot. In this case, it is well
known that the maximum streaming rate (the maximum reception rate
guaranteed to each peer) is approximately one for a large network
because the total upload bandwidth $N$ contributed by all peers (including
the source) has to be shared by $N-1$ peers (excluding the source)
\cite{Li05,Mundinger08 ,Kumar07}. Due to the limited communication
and computation overheads, we assume that each peer can only communicate
with a constant number of neighbors, which does not increase with
the network size. We assume that there is peer churn, so that the
topology is dynamic as new peers join or existing peers leave.

We now present our P2P streaming algorithm which consists of \emph{a
peer-pairing algorithm} and \emph{a chunk-dissemination algorithm}.
For convenience, we use the term chunk dissemination algorithm to
describe the joint neighbor selection and chunk selection algorithms
mentioned in the previous section. Our pairing algorithm is similar to the one in \cite{Law03}, except for the fact that we use directed edges. The fact that the edges are directed does not matter for adding or deleting nodes to the network; this part of our algorithm is identical to \cite{Law03}. However, the fact that the edges are directed and the fact that we are interested in achieving the maximum streaming capacity make our work quite different from \cite{Law03}, where the only goal is to construct an expander graph (with undirected edges) in a distributed fashion. But it is important to understand the pairing algorithm to proceed further. Therefore, we present it next.

\subsection{Peer Pairing Algorithm}

Under our peer-pairing algorithm, every peer has $M\geq2$ incoming
edges and $M$ outgoing edges as shown in Fig.~\ref{fig:basic_structure}.
We number the incoming edges of each peer as the first, second, ...,
$M$-th incoming edges of the peer and number its outgoing edges as
the first, second, ..., $M$-th outgoing edges. The peer where the
$m$-th outgoing edge ends is called \emph{the $m$-th child}, and
the peer where the $m$-th incoming edge begins is called \emph{the
$m$-th parent}. We assume that the $M$ outgoing edges of
a peer may end at the same peer, so that the number of  children
of a peer could be less than $M$. Similarly, the $M$ incoming edges
of a peer may begin at the same peer, so that the number of parents
of a peer could be less than $M$. Under our algorithm, every peer
receives chunks from its parents over its incoming edges and transmits
received chunks to its children over its outgoing edges. % Requires \usepackage{graphicx}
\begin{figure}[ht!]
\centering \includegraphics[width=0.9\mylength]{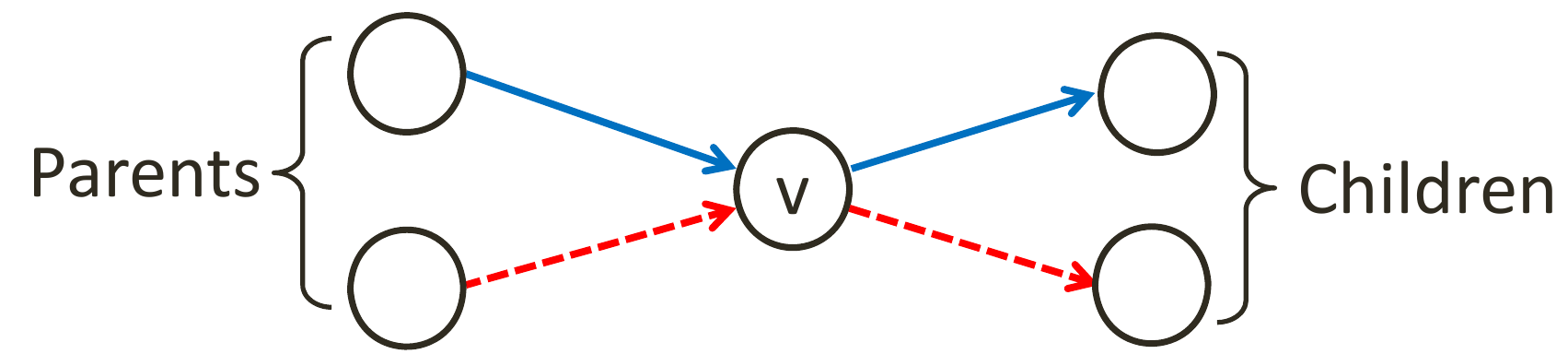}\\
 \caption{Basic structure of the peer-pairing algorithm for $M=2$: Every peer
maintains $M$ incoming edges and $M$ outgoing edges. The peers where
the outgoing edges end are called children, and the peers where the
incoming edges begin are called parents. Every peer $v$ receives
chunks from parents and transmits chunks to children.}

\label{fig:basic_structure}
\end{figure}

%Due to the limited communication and computation overheads, we assume that a peer can transmit chunks only to $M$ downstream peers, which we will call \emph{children}. Specifically, every peer maintains a list of $M$ children, which are indexed them from 1 through $M$ (i.e. the first child, the second child, ..., the $M$-th child.) We assume that the $M$ children in the list need not be unique, so that  some peers in the list could be the same. For example, if there are only two peers and $M=2$, one peer will be the first and the second children of the other peer, and \emph{vice versa}.  Similarly, we define \emph{parents} to be the peers that a peer can download chunks from. (Hence, each peer is a child of its parents.) Even though we do not impose a constraint on the number of parents that each peer maintains, we will show that  our pairing algorithm requires  only $M$ parents for each peer.

We next describe how every peer establishes its $M$ incoming and
$M$ outgoing edges.\\
 \textbf{Initially,} the network consists of only two peers, the source
peer (which we call peer 1 \tblue{throughout this paper}) and the first peer to arrive at the network
(which we call peer 2), as shown in Fig.~\ref{fig:initial_graph}.
Each peer establishes its first, second, ..., $M$-th outgoing edges
to the other peer, so that these edges are the first, second, ...,
$M$-th incoming edges of the other peer. Letting $V$ be the set
of current peers, we define $E_{m}$ to be the set of all $m$-th
edges, i.e., $E_{m}\triangleq\{(i,j)\in V^{2}|\text{ \ensuremath{j} is the \ensuremath{m}-th child of \ensuremath{i}}\},$
for $m=1,2,\cdots,M$. Initially, $E_{m}$ is given by $\{(1,2),(2,1)\}$
for all $m$ because there are only two peers. We define $\layer_{m}\triangleq(V,E_{m})$
to be the digraph consisting of the peer set $V$ and the $m$-th edges,
and call it \emph{layer} $m$ for $m=1,2,\cdots,M$. Layer $m$ represents
the pairing between every peer and its $m$-th child. By superposing
the $M$ layers, the current network topology can be expressed as
a multi-digraph $\layer^{*}=(V,E^{*})$, where $E^{*}$ is a multiset
defined as $E^{*}=\{E_{1}\cup E_{2}\cup,\cdots,E_{M}\}$.
\begin{figure}[ht]
\centering \subfigure[Initial graph: layer 1 (left), layer 2 (center), and the superposed graph (right)]{\label{fig:initial_graph}
\includegraphics[width=0.9\mylength]{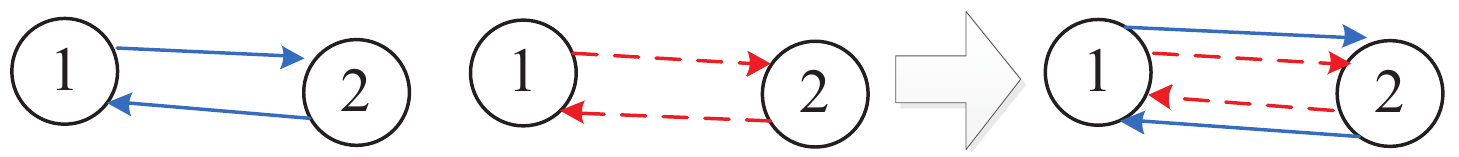}} \subfigure[Peer 4 arrives: layer 1 (left), layer 2 (center), and the superposed graph (right)]{\label{fig:peer_joining}
\includegraphics[width=0.9\mylength]{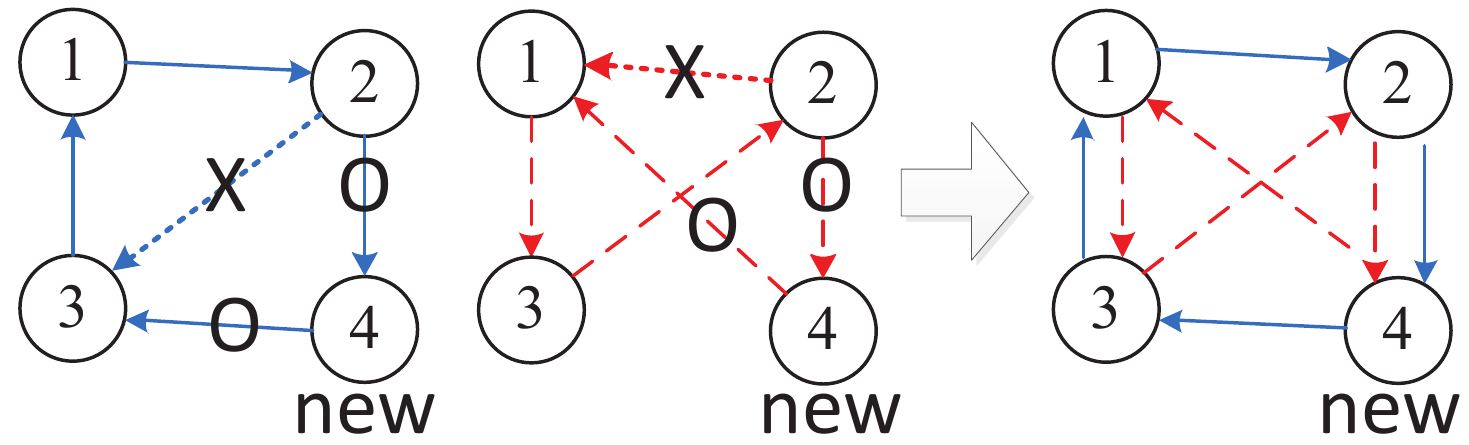}} \subfigure[Peer 2 leaves: layer 1 (left), layer 2 (center), and the superposed graph (right)]{\label{fig:peer_leaving}
\includegraphics[width=0.9\mylength]{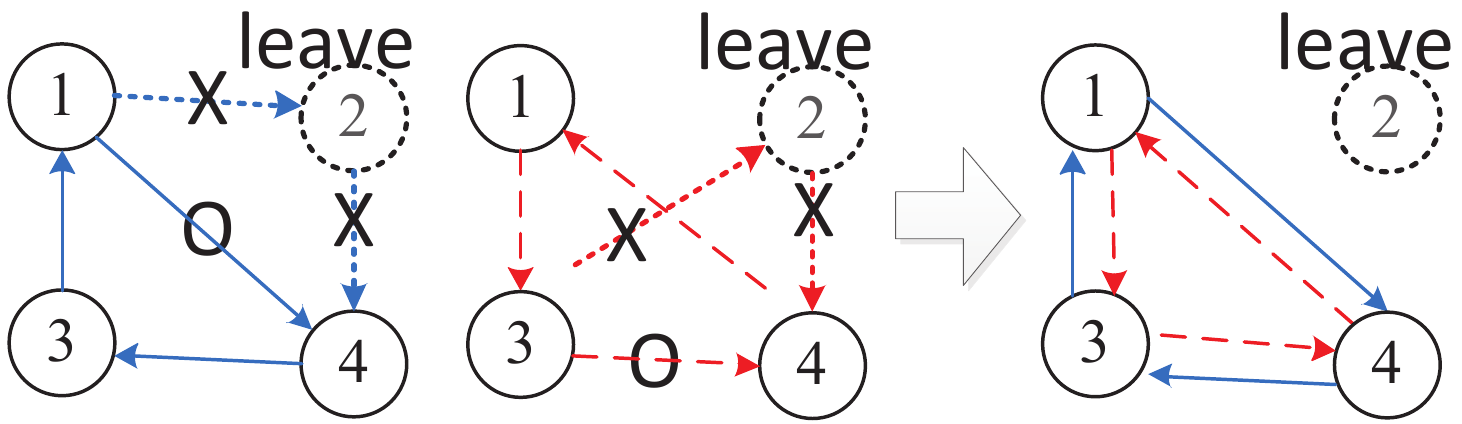}} \caption{Peer pairings for $M=2$: (a) Initially, peer 1 takes peer 2 as its
first and the second children, and \emph{vice versa}. (b) When peer
4 arrives to the network with existing peers 1, 2, and 3, it independently
chooses an edge (marked `x') from each layer uniformly at random and
breaks into the chosen edges. (c) When peer 2 leaves, its incoming
edge and outgoing edge (marked `x') in each layer are reconnected.}
\end{figure}

\textbf{When a new peer $v$ arrives,} this peer \emph{independently}
chooses an edge from each layer \emph{uniformly at random} and breaks
into the chosen edges as Fig.~\ref{fig:peer_joining}. Specifically,
if the peer $v$ arrives and randomly chooses $(p_{m},c_{m})\in E_{m}$
from layer $m$, the peer $v$ becomes a new $m$-th child of peer
$p_{m}$ and becomes a new $m$-th parent of peer $c_{m}$. Each layer
$m$ will then be updated as
\begin{align*}
V\Leftarrow & V\cup\{v\}\\
E_{m}\Leftarrow & E_{m}\cup\{(p_{m},v),(v,c_{m})\}\setminus(p_{m},c_{m}),\;\;\forall m.
\end{align*}
In practice, this edge-breaking can be easily implemented. If a new
peer arrives, it contacts a server to register its IP address. The
server then chooses $M(1+\alpha)$ IP addresses uniformly at random
with repetition and returns them to the peer. Here, $\alpha>0$ is
used in practice in case some peers are not reachable for some reason.
But for the purpose of analysis later, we assume that $\alpha=0$ and
all peers are reachable. Among these addresses, the peer contacts
$M$ reachable peers and breaks into their first, second, ..., $M$-th
outgoing edge, respectively.

\textbf{When an existing peer $v$ leaves,} its parents and children
will lose one of their neighbors as shown in Fig.~\ref{fig:peer_leaving}.
Let $p_{m}$ and $c_{m}$ be the parent and the child, respectively,
of peer $v$ in layer $m$, i.e., $(p_{m},v),(v,c_{m})\in E_{m}$.
(It is easy to see that every peer always has exactly
one parent and one child in each layer.) The parent $p_{m}$ in each
layer $m$ then directly contacts the child $c_{m}$ in the same layer
and takes the child as its new $m$-th child as shown in Fig.~\ref{fig:peer_leaving}.
In this case, the topology will change as follows:
\begin{align*}
V\Leftarrow & V\setminus\{v\}\\
E_{m}\Leftarrow & E_{m}\cup(p_{m},c_{m})\setminus\{(p_{m},v),(v,c_{m})\},\;\;\forall m,
\end{align*}
In practice, there is a chance that two or more successive ancestors of a peer in a layer leave the network simultaneously, which makes this edge-repairing impossible. This issue can easily be addressed by letting each peer remember the IP addresses of the several successive ancestors along the cycle in each layer. For the details, please refer to \cite{Kim11tech2}.

At any given time, the network topology $\layer^{*}$ that has been
constructed by the pairing algorithm satisfies the property stated
in the following lemma. The lemma and its proof are straightforward,
given the pairing algorithm, but we present them below to highlight
their importance to the analysis in the rest of the paper. \begin{lemma}\label{lemma:cycle_structure}
$\layer_{m}$ representing each layer $m$ is  a directed Hamiltonian cycle, i.e., every peer has
exactly one incoming and one outgoing edge in each layer, and all
the edges in $E_{m}$ form a single directed cycle. Hence, the superposed
graph $\layer^{*}$ is an $M$-regular multi-digraph, i.e., every
peer has exactly $M$ incoming edges and $M$ outgoing edges. \end{lemma}
\begin{proof} Initially, two peers form a single cycle in each layer.
When we add a new peer to each layer consisting of a single cycle,
the new peer simply breaks into an existing edge in a layer, maintaining
the existing cycle. When we remove a peer from each layer, its incoming
edge and its outgoing edge are reconnected, which also maintains the
cycle. Hence, when peers join or leave, the cycles built initially
do not vanish, but they expand or shrink in size. Hence, each layer
is always a cycle graph. %Since every peer has exactly one incoming and one outgoing edge in each layer, it has exactly $M$ incoming and
%$M$ outgoing edges in the superposed graph $G^*$.
\end{proof}

Lemma~\ref{lemma:cycle_structure} implies that, under the pairing
algorithm, each peer needs to communicate only with $M$ parents and
$M$ children. Hence, the communication/computation overhead to maintain
multiple TCP or UDP sessions does not increase with the network size.

\emph{Remark:} The pairing algorithm is fully distributed,
except for the information provided by a central server to identify
a few other peers in the network.
The server only maintains the list of registered peers and
their IP addresses, which need not be updated frequently. If some
peers in the list do not exist in the network any more, the server
may send the IP address of such a peer to a new peer. However, this
does not affect the pairings of the new peer because the new peer
will only contact $M$ reachable addresses among the $M(1+\alpha)$
received addresses. A central server to perform such minimal functionalities
is usually called a tracker, and is used by most P2P networks. In
our analysis, we do not consider the details of the information sharing
between the tracker and the peers. We simply assume that $M$ random
addresses are provided to a new peer to enable it to execute the pairing
algorithm.

\subsection{Chunk Dissemination Algorithm}

\label{subsec:scheduling_and_chunk_algorithm}

While the pairing algorithm determines the network topology, the chunk
dissemination algorithm determines how chunks are disseminated over
a given topology. We here present our chunk dissemination algorithm
that can provide provable throughput and delay bounds.

%Once a schedule is chosen for the peers, we must decide which chunk to disseminate over the selected outgoing edge at a peer in each timeslot.} For this purpose, we assume that chunks are ``colored'' with one of $K-1$ \tblue{predetermined} colors.

Assume that the source generates at most one chunk during
every timeslot, except timeslots $0,K,2K,\cdots$ for some integer $K>2$. Since at most $K-1$ chunks are generated during  every $K$ timeslots,  the maximum chunk-generating rate under our algorithm is $(1-1/K)$. We call the chunk generated at timeslot $t$ \emph{chunk t}. Suppose there are $K-1$ predetermined colors, numbered color 1, color 2, ..., color $K-1$ and we color each chunk $t$ with color $(t\mod K).$ In other words, the chunks are colored from $1$ through $K-1,$ and then again starting from $1,$ with the process repeating forever. We call the chunk with color $k$ simply \emph{a color-$k$ chunk}. If a color-$k$ chunk is generated at time $t$ at the source, then color $k$ chunks are also generated at time $t+K,$ $t+2K$ and so on. If chunks are not generated periodically in this manner, then a smoothing buffer has to be used at the source to ensure that only $K-1$ chunks are periodically generated for every $K$ timeslots, and any other additional chunks are stored for later transmission. Thus, there will be a queueing delay at the source for storing the additional chunks which we ignore since our goal here is to characterize the scaling behavior of the end-to-end transmission delay from the source to all peers as a function of $N$.

Recall that every peer can upload at most one chunk to one other peer
in a timeslot. At the beginning of each timeslot, every peer \emph{schedules}
one of its outgoing edges, i.e., the peer selects an outgoing edge
and uploads a chunk over that edge. Specifically, every peer $i$
shares the same scheduling vector $\Lambda=(\lambda_{1},\lambda_{2},\cdots,\lambda_{K})$,
where $\lambda_{k}\in\{1,2,\cdots,M-1\}$ for all $k<K$ and $\lambda_{K}=M$.
Peer $i$ schedules its outgoing edges, cycling through the elements
in the scheduling vector. For example, if $M=3$ and $\Lambda=(1,2,1,3)$,
every peer repeats scheduling its first, second, first, and third
outgoing edges sequentially. \emph{We note that the scheduling round
of a peer need not be synchronized with the other peers, i.e., at
a given timeslot, peers may schedule different types of outgoing edges.}

Suppose peer $v$ schedules the $k$-th edge in the scheduling
vector (i.e., the outgoing edge in layer $\lambda_{k}$ or equivalently
the $\lambda_{k}$-th outgoing edge) at the beginning of timeslot
$t$. Let $Q_{v,k}(t)$ be the set of the color-$k$ chunks that peer
$v$ has received before timeslot $t$ for $0<k<K$. If $k<K$, peer $v$ chooses the chunk from $Q_{v,k}(t)$ that was generated most recently (called \emph{the latest chunk}) and uploads this chunk over the scheduled edge, regardless of whether or not the other end possesses the chunk. If $k=K$, the peer transmits the latest chunk in $Q_{v,\mu(v)}(t)$,
where $\mu(v)$ is a random variable uniformly chosen from $\{1,2,\cdots,K-1\}$
when the peer joined the network.  We also assume that $\mu(v)$ does not change once it is determined.
We call $\mu(v)$ \emph{the
coloring decision of peer $v$.} During $K$ timeslots of a scheduling
round, peer $v$ will transmit the latest chunks with color 1, color
2, ..., color $K-1$, and color $\mu(v)$ over the $\lambda_1$-st, $\lambda_2$-nd, $\cdots$, $\lambda_{K}$-th outgoing edges, respectively. Since the
scheduling rounds of peers are asynchronous, peers may transmit chunks
with different colors at a given timeslot. Note that when a peer
receives a chunk, this chunk will be unavailable for uploading till
the next timeslot. Furthermore, we have assumed
implicitly that only the latest chunk with each color is available for uploading at a peer. Thus, if a color-$k$ chunk that is generated later than the latest chunk in $Q_{v,k}(t)$ of peer  $v$ arrives at  peer $v$ at timeslot $t$, the peer
will not upload all the  chunks received before timeslot $t$. We will show later that all chunks are delivered
to all the peers despite the fact that older chunks are discarded.
In other words, we will prove that the older chunks have already been disseminated by a peer by the time they are discarded and so are no longer necessary
from the point of view of data dissemination (although they may be
retained for playout at the peer).

We have presented our chunk-dissemination algorithm running on top
of the pairing algorithm. Besides our algorithm, other chunk dissemination
algorithms, such as \emph{the random useful algorithm}  \cite{Massoulie07},
\emph{the latest-blind algorithm}  \cite{Sanghavi07}, and \emph{the
latest-useful algorithm} \cite{Bonald08}, can be potentially used
over the network topology that is constructed by the pairing
algorithm. Our performance analysis is, however, only for the chunk
dissemination algorithm proposed here.

\subsection{Bounds on Streaming Rate and Delay}

Our P2P algorithm will be evaluated using two metrics: streaming rate
and delay.

\emph{Streaming Rate:} What is the streaming rate achieved by our
P2P algorithm? The streaming rate is defined as the chunk reception
rate guaranteed to all peers. When peers contribute unit bandwidth,
the total upload bandwidth $N$ contributed by all peers (including
the source) has to be shared by $N-1$ peers (excluding the source).
Thus, the download bandwidth per peer cannot exceed $\frac{N}{N-1}$,
which is approximated to one for large $N$. Hence, the optimal streaming
rate is close to one for a large network. In Section~\ref{sec:Througput_Analysis},
we will show that our algorithm disseminates all the chunks to all
peers, and achieves a streaming rate of $1-\frac{1}{K}$, which is
arbitrarily close to the optimal streaming rate for sufficiently large
$K$.

\emph{Dissemination Delay:} What is the delay that can be achieved
by our P2P algorithm? When each peer is allowed to disseminate chunks
only to a constant number of neighbors, as in a real P2P topology,
the fundamental limit of the delay required to disseminate a chunk
to all peers is known to be $\Omega(\log N)$.%
\footnote{If we trace the paths that a chunk has been transmitted, the paths
form an arborescence with a bounded degree rooted at the source. Since
this arborescence has at least $\Omega(\log N)$ depth, distributing
a chunk to all peers requires at least $\Omega(\log N)$ transmissions.%
} This limit is a lower bound on the delay to disseminate multiple
chunks because the contention between multiple chunks at a peer can
only increase the dissemination delay. In Section~\ref{sec:Delay_Analysis},
we show that our algorithm achieves this fundamental limit, i.e.,
every chunk arriving at the source at rate $(1-\frac{1}{K})$ is disseminated
to all peers within $\Theta(\log N)$ timeslots with high probability
under our algorithm.

\section{Throughput and Delay Analysis}

\label{sec:Througput_Analysis}

In this section, we show that our algorithm achieves the streaming
rate of $(1-\frac{1}{K})$, i.e., each chunk arriving at the source
at rate $(1-\frac{1}{K})$ can be disseminated to all peers by our
algorithm.
To this end, we first characterize the graph over which color-$k$ chunks are disseminated. We then show that no color-$k$ chunks are dropped before being disseminated to all peers.

As described in Section~\ref{subsec:scheduling_and_chunk_algorithm},
during every scheduling round of a peer $v$, peer $v$ transmits
the latest color-$1$ chunk, color-$2$ chunk, ..., color-$(K-1)$
chunk, and color-$\mu(v)$ chunk over its $\lambda_{1}$-st, $\lambda_{2}$-nd,...,
$\lambda_{(K-1)}$-st, $M$-th outgoing edges, respectively. Thus,
color-$k$ chunks are delivered over the $\lambda_{k}$-th outgoing
edges from all peers (i.e., the edges in layer $\lambda_{k}$) and
the $M$-th outgoing edges from peers $v$ with $\mu(v)=k$. If we
define flow graph $G_{k}$ ($k=1,2,\cdots,K-1$) to be the graph consisting
of the edges carrying color-$k$ chunks, the flow graph can be expressed
as a multi-digraph $G_{k}=(V,E_{\lambda_k}\cup\mathcal{E}_{M,k})$, where
\begin{equation}
\mathcal{E}_{M,k}=\{(i,j)\in E_{M}|\mu(i)=k\}.\NN
%\label{eq:definition_of_E_mk}
\end{equation}
Thus, color-$k$ chunks are disseminated over flow graph $k$, where the out-degree of every peer is at most two. (See
the example of the flow graphs for $K=3$ and $M=2$ in Fig.~\ref{fig:flow_graphs}).
\begin{figure}[ht]
\centering \subfigure[Layer 1 (left) and  layer 2 (right): The numbers on  peers $v$ in layer 2 are the coloring decisions $\mu(v)$.]{\label{fig:layers}
\includegraphics[width=0.9\mylength]{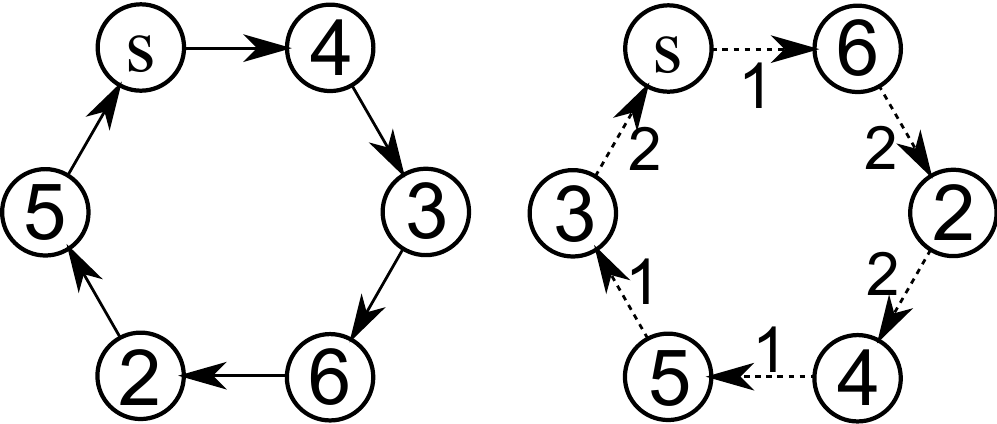}} \subfigure[Flow graphs $G_1$ (left) and $G_2$ (right)]{\label{fig:flow_graphs}
\includegraphics[width=0.9\mylength]{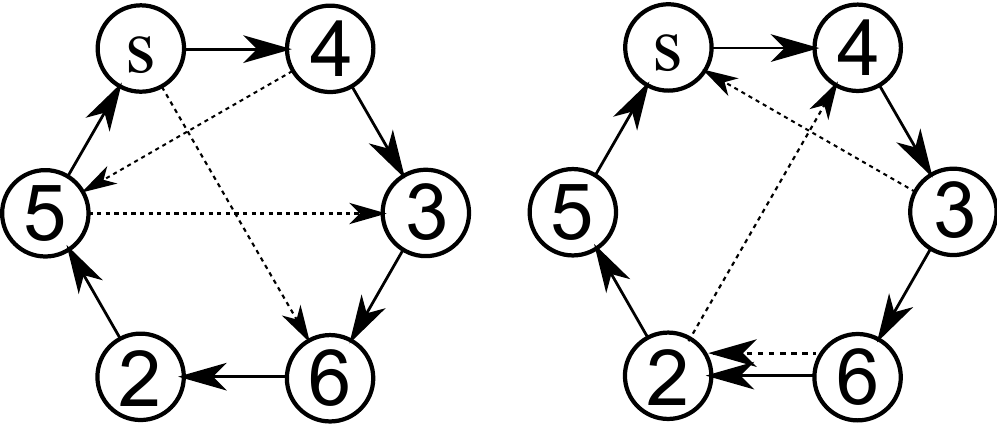}} \caption{A network with 6 peers named \tblue{1} (source), 2, 3,..., 6,  where $M=2$,
$K=3$, $\Lambda=(1,1,2)$ and $(\mu(s),\mu(2),\cdots,\mu(6))=(1,2,2,1,1,2)$.
Flow graph $G_{1}$ consists of the edges in the first layer and the
edges $(i,j)$ with $\mu(i)=1$ in the second layer. Similarly, flow
graph $G_{2}$ consists of the edges in the first layer and the edges
$(i,j)$ with $\mu(i)=2$ in the second layer.}
\end{figure}

We next study how color-$k$ chunks are disseminated over flow graph
$G_{k}$. Recall that if a color-$k$ chunk is generated at timeslot $t$, chunks $t+K,t+2K,\cdots$ are all of color $k$.
 We call these chunks \emph{later chunks of chunk $t$}. Since our chunk dissemination algorithm transmits only the latest chunk of each color,  if a peer has received both chunk $t$ and a later chunk,
the peer will not transmit chunk $t$ any longer. Thus, if all the
peers that have received chunk $t$ have also received a later chunk,
chunk $t$ cannot be disseminated to the remaining peers. However,
the following proposition shows that this scenario does not occur.
\begin{proposition}\label{prop:dissemination_over_flowgraphs} Under
our algorithm, if a peer receives chunk $t$ during timeslot $l\geq t$,
this peer has received chunk $t-K$  no later than timeslot $l-K$.
\end{proposition}
\begin{proof} Without loss of generality, fix
$t=0$. We prove by induction that if a peer receives chunk $0$ (with
color $k$) during timeslot $l\geq0$, it has received chunk $-K$
before or during timeslot $l-K$.

Initially $(l=0)$, chunk 0 arrives at the source during timeslot
0. Since the coloring queue moves one color-$k$ chunk every $K$
timeslots, chunk $-K$ must have arrived at the dissemination queue
of the source during timeslot $-K$. Since the source is the only
peer that has chunk 0 during timeslot 0, the statement is true for
$l=0$.

We now assume that the statement is true for $l\leq t'-1$. We next
show that the statement is also true for $l=t'$. Consider a particular
peer $j$ that receives chunk 0 for the first time through an incoming
edge $(i,j)$ during timeslot $t'$. This implies that peer $i$ has
received chunk 0 for the first time during timeslot $t''\in\{t'-K,t'-K+1,\cdots,t'-1\}$.
(If $t''\geq t'$, peer $j$ cannot receive chunk 0 during timeslot
$t'$. If $t''<t'-K$, peer $j$ must have received chunk 0 during
timeslot $t'-K$.) By the induction hypothesis, peer $i$ must have
received chunk $-K$ before or during timeslot $t''-K$. Thus, chunk
$-K$ has been the latest color-$k$ chunk to peer $i$ from timeslot
$t''-K+1$ to timeslot $t''$.
Since the schedule is cyclic, edge $(i,j)$
 was scheduled during timeslot $t'-K$. Since $t''-K+1\leq t'-K\leq t''$,
chunk $-K$ must have been transmitted during that timeslot over $(i,j)$.
Thus, the statement is true for $l=t'$.

By induction, if a peer first receives chunk 0 during timeslot $l$,
it has received chunk $-K$ before or during timeslot $l-K$. \end{proof}
Proposition~\ref{prop:dissemination_over_flowgraphs} implies that
if a peer receives chunk $t$, it has at least $K$ timeslots (one
scheduling round) to distribute the chunk to its children before a
later chunk arrives. Since the peer schedules each outgoing edge in
$G_{k}$ exactly once during every $K$ timeslots for transmitting
color-$k$ chunks, the peer will transmit chunk $t$ to its children
in $G_{k}$ before a later chunk arrives. Thus, every color-$k$ chunk
arriving at the source can be disseminated to all the peers that are
connected from the source in $G_{k}$, i.e., there exists a path from
the source to the peers in $G_{k}$. Since every flow graph $G_{k}$
contains layer $\lambda_{k}$, which is a Hamiltonian cycle, every
peer is connected from the source. Thus, all chunks arriving at rate
$(1-1/K)$ can eventually be disseminated to all peers under our chunk-dissemination
algorithm.

In streaming applications, this throughput analysis is meaningless
without a delay guarantee. We next consider how fast each color-$k$
chunk is disseminated over flow graph $G_{k}$. \begin{lemma}\label{lemma:dissemination_speed}
Let $d_{k}(v)$ be the shortest distance from the source to peer $v$
in flow graph $G_{k}$. If a color-$k$ chunk arrives at  the source during timeslot $t$, peer $v$ receives the chunk
by timeslot $t+Kd_{k}(v)$. \end{lemma}
\begin{proof}
If a peer receives a color-$k$ chunk, it transmits the chunk to its children
in $G_{k}$ during next $K$ timeslots by Proposition~\ref{prop:dissemination_over_flowgraphs}.
Thus, the time until peer $v$ receives the chunk through the shortest
path in $G_{k}$ from the source to peer $v$ does not exceed $Kd_{k}(v)$.
\end{proof}

Lemma~\ref{lemma:dissemination_speed} shows that the delay required
to disseminate a color-$k$ chunk to peer $v$ is upper bounded by
$K\cdot d_{k}(v)$ timeslots. If we call the maximum distance $\max_{v\in V}d_{k}(v)$
\emph{the depth $d_{k}^{*}$ of $G_{k}$}, the delay to disseminate
a color-$k$ chunk to all peers is upper bounded by $K d_{k}^{*}$
timeslots. Thus, if $d_{k}^{*}$ is $\Theta(\log N)$ for all $k=1,2,\cdots,K-1$,
the dissemination delay of our algorithm is upper bounded by $\Theta(\log N)$
timeslots. In the next section, we prove that this is true with high
probability.

\section{Depth of Flow Graphs}
\label{sec:Delay_Analysis}

In this section, we consider the depth $d_{k}^{*}$ of each flow graph
$G_{k}$. Since all the layers are random graphs, which are affected
by the history of past peer churn, the corresponding flow graph is
also a random graph. Thus, the depth $d_{k}^{*}$ of the flow graph
must also be a random variable. The objective of this section is to
show the following proposition:
\begin{proposition}\label{prop:depth_of_flowgraph}
For any $\psi\in (0, q/2)$, the maximum distance $d_k^*$ from the source to all other peers in flow graph $G_{k}$ is $O(\log_{1+\psi}N)$ with probability $1-O(\log_{1+\psi}N/N^{\sigma'})$ for some positive constant $\sigma'$ \tblue{and $q=1/(K-1)$.}
 \end{proposition}

To prove Proposition~\ref{prop:depth_of_flowgraph}, we follow the
following three steps. First, to characterize random variable $d_{k}^{*}$,
we need to characterize the random graph $G_{k}$. In this step, we
show that there is an alternative way to construct the random graph
which is stochastically equivalent to the construction described in
Section~\ref{sec:System_Model}. In the second step, using the alternative
construction, we will show that the number of peers within $l$ hops
from the source in $G_{k}$ increases exponentially in $l$ until the number is no larger than $N/2$. In the last step, we show that the number of remaining peers that are not within $l$ hops from the source reduces exponentially.

\subsection{Distribution of Flow Graphs}
\label{subsection:dist_flow_graphs}

%\subsection{Distribution of Layers}
%\begin{lemma}\label{lemma:equal_incoming_outgoing}
%For every $G\in \cC(V)$ and every $S\subset V$, we have $I_{G}(S)=I_{G}(S^C)$.
%\end{lemma}
Consider two random multi-digraphs $G'=(V,E')$ and $G''=(V,E'')$
that have the same peer set $V$ and random edges. For every possible
multi-digraph $G$ with peer set $V$, we say that these two random
graphs have the same distribution if $\prob{G'=G}=\prob{G''=G}$.
In this subsection, we consider how flow graph $G_{k}$ is distributed
and how to construct a random graph that has the same distribution
as flow graph $G_{k}$.

Recall that flow graph $G_{k}$ is the superposed graph of layer $\lambda_{k}$
and $(V,\mathcal{E}_{M,k})$, a subgraph of layer $M$. Thus, we first
consider the distribution of each layer. As in the first example of
Section~\ref{sec:examples}, construct a random Hamiltonian cycle
by permuting the peers in $V$ and let $H$ denote this cycle.
(Recall that all the graphs in this paper are directed graphs, and thus we omit repeatedly mentioning ``directed.'') Let
$\cC(V)$ be the set of all possible Hamiltonian cycles that we can
make with peer set $V$. Then, it is easy to see that $H$ is distributed
as
\begin{equation}
\prob{H=G}=\frac{1}{|\cC(V)|}=\frac{1}{(N-1)!},\label{eq:distribution_hamilton}
\end{equation}
for every $G\in\cC(V)$. The following proposition shows each layer
$L_{m}$ has the same distribution as $H$. \begin{proposition}\label{prop:layer_distribution}
For $m=1,2,\cdots,M$, each layer $m$ (denoted by $L_{m}$) and random
 Hamiltonian cycle $H$ have the same distribution, i.e., for every $G\in\cC(V)$,
\[
\prob{L_{m}=G}=\prob{H=G}=\frac{1}{(N-1)!}.
\]
\end{proposition} \begin{proof} We prove by induction. Fix layer
$m$ ($L_{m}$). Initially (with two peers 1 and 2), layer $m$ is
$1\rightarrow2\rightarrow1$, which is the only possible
Hamiltonian cycle with two peers. Thus, each layer and $H$ have the
same distribution.

Suppose that $L_{m}=(V,E_{m})$ and $H\in\cC(V)$ have the same distribution
for peer set $V$ with $|V|\geq2$. If a new peer $i\notin V$ joins,
this peer chooses one edge from $L_{m}$ uniformly at random and
breaks into the edge. Let $L_{m}'$ be layer $m$ after adding peer
$i$. Similarly, choose one edge from $H$ uniformly at random and
add peer $i$ into the edge. Let $H'$ be the Hamiltonian cycle after
adding peer $i$ to $H$. Clearly, $L_{m}'$ has the same distribution
as $H'$. It is easy to see that making a Hamiltonian cycle by permuting
$|V|$ peers and then adding peer $i$ into the cycle is equivalent
to making a Hamiltonian cycle $H''$ by permuting the peers in $V\cup\{i\}$.
Thus, $H'$ and $H''$ have the same distribution, and so do $L_{m}'$
and $H''$.

Suppose that $L_{m}=(V,E_{m})$ and $H\in\cC(V)$ have the same distribution
for peer set $V$ with $|V|>2$. We remove a peer $i\in V$ from layer
$m$ and connect its incoming and outgoing edges. Let $L_{m}'$ denote
the layer after removing the peer. Since $L_{m}$ and $H$ have the
same distribution, if we remove peer $i$ from $H$, the resulting
graph $H'$ will also have the same distribution as $L_{m}'$. It
is easy to see that making a Hamiltonian cycle by permuting the peers
in $V$ and removing peer $i$ is equivalent to making a Hamiltonian
cycle $H''$ by permuting the peers in $V\setminus\{i\}$. Thus, $H'$
and $H''$ have the same distribution, and so do $L_{m}'$ and $H''$.

By induction, at any given time, each layer $L_{m}$ with peer set
$V$ has the same distribution as a random Hamiltonian cycle obtained
by permuting the peers in $V$. \end{proof} %The proof is provided in Appendix~\ref{proof:prop:layer_distribution}.

Proposition~\ref{prop:layer_distribution} shows that conditioned
on peer set $V$, each layer has the same distribution as $H$. Among
layer $\lambda_{k}$ and $(V,\mathcal{E}_{M,k})$ that form flow graph
$G_{k}$, $\lambda_{k}$ can be replaced with $H$ for analysis.

We next consider how $(V,\mathcal{E}_{M,k})$ is distributed. Note
that $(V,\mathcal{E}_{M,k})$ is a subgraph of layer $M$ that consists
of only the edges $(i,j)\in E_{M}$ with $\mu(i)=k$. Since each peer
$i$ makes its coloring decision to be $\mu(i)=k$ with probability
$q\triangleq1/(K-1)$, $(V,\mathcal{E}_{M,k})$ can be seen as the
graph made from layer $M$ by independently removing each edge with
probability $1-q$. Since layer $M$ has the same distribution as $H$,
$(V,\mathcal{E}_{M,k})$ has the same distribution as the graph $H'$
that we obtain from $H$ by removing each edge with probability $1-q$.

Note that conditioned on peer set $V$, layer $\lambda_{k}$ and layer
$M$ are mutually independent because peer pairing in a layer has
been independent from that in another layer, i.e., for any $\mathcal{G'},\mathcal{G''}\subset\cC(V)$,
\[
\prob{L_{\lambda_{k}}\in \mathcal{G'},L_{M}\in \mathcal{G''}}=\prob{L_{\lambda_{k}}\in \mathcal{G'}}\prob{L_{M}\in \mathcal{G''}}.
\]
Thus, the graph that has the same distribution as $G_{k}$ can be
constructed from two independent random Hamiltonian cycles as follows:
\begin{proposition}\label{prop:distribution_of_flowgraph} Construct
two random Hamiltonian cycles $H_{1}$ and $H_{2}$ by permuting the
peers in $V$ independently for each. Remove each edge from $H_{2}$
with probability $1-q$, where $q=1/(K-1)$, and call the resulting
graph $H_{2}'$. If we superpose $H_{1}$ and $H_{2}'$, the superposed
graph $H^{*}$ has the same distribution as flow graph $G_{k}$ for
$k=1,2,\cdots,K-1$. \end{proposition}

Note that $H^{*}$ in Proposition~\ref{prop:distribution_of_flowgraph}
is identical to the graph in the second example of Section~\ref{sec:examples}.
\emph{Thus, the maximum distance from a given node
to all other nodes in that graph is stochastically equivalent to the
depth of each flow graph.} By proving the depth of $H^{*}$ is $\Theta(\log N)$
with high probability, we show that the depth of flow graph is also
$\Theta(\log N)$ with high probability.

There exist several ways to construct  $H_1$ and $H_2'$. The simplest way is to permute the peers and connect this permutation of peers with edges. However, when we construct $H_{1}$ and $H_{2}'$ using this method, analyzing the depth of $H^{*}$ is not straightforward. Instead, we use another equivalent process which provides us with a tractable
construction amenable to analysis:\\
 \textbf{Flow Graph Construction \tblue{(FGC)} Process:} Given peer set  $V$, \tred{and coloring decision ${\mu}(\cdot)$,}
\begin{enumerate}
\item \tblue{$v_1=1$ (source),} $Z=\{v_1\}, E^{(1)}=E^{(2)}=\emptyset$, $t=1$ (Here, the variable $t$ is used to indicate that there are $t-1$ edges in $E^{(1)}$), \tblue{and
    $$\vec \tau \triangleq (\tau_1, \tau_2,\cdots, \tau_N),$$
    where $\tau_1, \tau_2,\cdots, \tau_N$ are independent Bernoulli random variables with the same mean $q=1/(K-1)$.}
\item \tblue{Start iteration $t$: we will draw outgoing edges from $v_t$.}
    \tred{At iteration $t\geq 1$, among the peers in
    $Z\setminus\{v_{l}\;|\; l<t\}$,
let $v_t$ be the peer that was \tred{first} added to $Z$.  (When $t=1$, the set $\{v_l | l < t\}$ is taken to be the empty set.)}
\item Choose $c_{t}$ \tred{and $c_{t}'$} from $C(v_{t},E^{(1)})$ \tred{and $C(v_{t},E^{(2)})$,
respectively,} uniformly at random, where $C(v,E)$ is the set of peers satisfying
\begin{condition}\label{condition:to_be_a_cycle} For every $c\in C(v,E)$, \\
 (a) There is no edge ending at peer $c$ in $E$.\\
 (b) Adding edge $(v,c)$ to graph $(V,E)$ does not incur a loop or a cycle unless the cycle is Hamiltonian.
\end{condition}
\item Add \tred{$c_t$ to $Z$ and} $(v_t,c_t)$ to $E^{(1)}$.
\item \tblue{If $c_t\notin Z$, add $c_t$ to $Z$ and then let $v_{|Z|}=c_t$.}
%\item Define a variable $\tred{\mu}_{t}$ as follows: with probability $q$, let $\tred{\mu}_{t}=1$
%and with probability $1-q,$ let $\tred{\mu}_{t}=0.$
\item If \tblue{$\tau_t=1$,} \tred{${\mu}(v_t)=k$,}
\tblue{\begin{enumerate}
\item Choose $c_t'$ from $C(v_{t},E^{(2)})$ uniformly at random.
\item Add  $(v_t, c_t') $ to $E^{(2)}$.
\item If $c_t'\notin Z$, add $c_t'$ to $Z$ and then let $v_{|Z|}=c_t'$.
\end{enumerate}}
 \tblue{If $\tau=0$, no edge is added to $E^{(2)}$, and no peer is added to $Z$ in this step. In this case, we set $c_t'=\infty$.}
\item If $t<N$, increase $t$ by one and go to Step~2.
\item Return $E^{(1)}$ and $E^{(2)}$.
\end{enumerate}

For given peer set $V$, we can construct two random graphs $(V, E^{(1)})$ and $(V, E^{(2)})$. The next proposition shows that if we superpose these graphs, the resulting graph has the same distribution as $H^*$.
\begin{proposition}\label{prop:flow_graph_process}
Random graphs $(V, E^{(1)})$ and $(V, E^{(2)})$ constructed by the FGC process are mutually independent and have the same distribution as $H_1$ and $H_2'$, respectively.
\end{proposition}

\tblue{
In the rest of this subsection, we provide the intuition of the proof. The detailed proof is provided in Appendix~\ref{proof:prop:flow_graph_process}. When we construct $(V,E^{(1)})$ using the FGC process, we iteratively pick a peer $v_t$  that does not have an outgoing edge and draw an edge from it to a random peer $c_t$ that does not incur a non-Hamiltonian cycle. Thus, after drawing $N$ edges in this manner, the resulting graph $(V,E^{(1)})$ will be a Hamiltonian cycle in $\cC(V)$. Since we have chosen $c_t$ uniformly at random among the candidates not incurring a non-Hamiltonian cycle, the resulting Hamiltonian cycle $(V,E^{(1)})$ is uniformly distributed in $\cC(V)$.}

\tblue{We now consider how $(V,E^{(2)})$ has the same distribution as $H_2'$. Recall that we have obtained  $H_2'$ by independently removing each edge with probability $1-q$ from a random Hamiltonian cycle. Hence, if we draw a random Hamiltonian cycle as we have drawn $(V,E^{(1)})$ and then remove each edge beginning at peer $v_t$ with $\tau_t=0$, the resulting graph  should have the same distribution as $H_2'$. Say the resulting graph $H'$. Instead of removing edges after completing the random Hamiltonian cycle, we now draw a random edge from each peer $v_t$ with $\tau_t=1$, as we did for $(V,E^{(1)})$, and stop drawing once we finish drawing edges from the peers. Say the resulting graph $H''$.
Then, $H''$, which we have drawn skipping some edges, should have the same distribution as $H'$, which we have drawn deleting some edges from a random Hamiltonian cycle. Since the process of drawing $H''$ is identical to the way how the FGC process constructs $(V,E^{(2)})$, both   $(V,E^{(2)})$ and $H_2'$ have the same distribution.}

\tblue{We finally show that $(V,E^{(1)})$ and $(V, E^{(2)})$ constructed by the FGC process are mutually independent. At each iteration, we have chosen the children  $c_t$  and $c_t'$ of peer $v_t$ independently of each other. Further, $\tau_t$ is chosen independently of $c_t$ and $c_t'$. Hence, after $N$ iterations, $E^{(1)}$ and $E^{(2)}$ are mutually independent, and thus so are the resulting graphs $(V,E^{(1)})$ and $(V, E^{(2)})$.}

\tred{In the rest of this subsection, we prove this proposition. We first show that $(V,E^{(1)})$ and $(V, E^{(2)})$ are mutually independent. We then show that $(V,E^{(1)})$ and $(V, E^{(2)})$ have the same distribution as $H_1$ and $H_2'$, respectively.
At each iteration, the choice of $c_t$ is independent of $E^{(2)}$, and the choice of $c_t'$ is independent of $E^{(1)}$.  Thus,  $E^{(1)}$ and $E^{(2)}$ at the end of the process are mutually independent, and so are $(V,E^{(1)})$ and $(V, E^{(2)})$.}

\tred{We next show that  $(V,E^{(1)})$ and $(V, E^{(2)})$ have the same distribution as $H_1$ and $H_2'$, respectively. To this end, we need to find the number of candidates for $c_t$ and $c_t'$.
\begin{lemma}
At iteration $t$ of the FGC process, we have $$|C(v_t,E^{(1)})|=N-|E^{(1)}|-1=N-t,$$
$$|C(v_t,E^{(2)})|=N-|E^{(2)}|-1=N-\summu{t-1}-1,$$
for $t<N$.
\end{lemma}
Using Lemma~\ref{lemma:num_of_candidates}, we   finish the proof of  Proposition~\ref{prop:flow_graph_process}.
At each iteration $t<N$, the number of possible choices for $c_t$ is $N-t$. At iteration $t=N$, the must be only one acyclic path that passes all peers because we have drawn $N-1$ edges and the edges do not form a cycle by Condition~\ref{condition:to_be_a_cycle}. Thus, there must be only one choice for $c_t$. Overall, there exist $(N-1)!$ combinations for $(c_1, \cdots, c_N)$. Since we have chosen $c_t$ uniformly at random, $(c_1, \cdots, c_N)$ is uniformly distributed among the combinations. Note that each $(c_1, \cdots, c_N)$ corresponds to a unique Hamiltonian cycle $(V, E^{(1)})$. Thus, $(V, E^{(1)})$ is uniformly distributed in $\mathcal{C}(V)$ as  $H_1$.}

\tred{We next show that $(V,E^{(2)})$ has the same distribution as $H_2'$. Fix $\mu(v)$ for each $v\in V$. Note that there exists a unique subsequence  $t_1<t_2<\cdots< t_{J}$ of $1,2,\cdots, N$ that satisfies  $\mu(v_t)=k$ if $t\in\{t_i\;|\; 1\leq i\leq J\}$ and $\mu(v_t)\neq k$  otherwise. By definition, it follows that $J=\summu{ N}$.
By Lemma~\ref{lemma:num_of_candidates}, at the beginning of each iteration $t_i$, the number of possible choices for $c_{t_i}'$  is $N-i$ because $E^{(2)}$ at this moment has $i-1$ edges. Hence, there exists $(N-1)!/ (N-J-1)!$ combinations for $(c_{t_1}', c_{t_2}', \cdots, c_{ t_{J}}')$. Since $c_{t_i}'$ is chosen uniformly at random at each iteration $t_i$, $(c_{t_1}', c_{t_2}', \cdots, c_{ t_{J}}')$ is also uniformly distributed among these combinations. Since $(V, E^{(2)})$ for each $(c_{t_1}', c_{t_2}', \cdots, c_{ t_{J}}')$ corresponds to a unique subgraph of a Hamiltonian cycle where only peers $v_{t_1},\cdots, v_{t_J}$ have an outgoing edge. It is easy to see that the number of such subgraphs is also $(N-1)!/ (N-J-1)!$. Thus, $(V, E^{(2)})$ is uniformly distributed among all such subgraphs. Recall that we obtain $H_2'$ by removing each edge $(i,j)$ with $\mu(i)=0$ from $H_2$. Hence, conditioned on the same $\mu(\cdot)$, $H_2'$ must be one of the subgraphs. Since $H_2$ is uniformly distributed in $\mathcal{C}(V)$, it is easy to show that $H_2'$ is also uniformly distributed among the subgraphs. Thus, conditioned on $\mu(\cdot)$, $(V, E^{(2)})$ and $H_2'$ have the same distribution.}

Overall, $(V,E^{(1)},E^{(2)})$ that we have drawn using  the FGC process will have the same distribution as $H^{*}$ by Proposition~\ref{prop:distribution_of_flowgraph}.
Thus, the FGC process can be seen as an another way to construct flow graph $G_{k}$. Note that we do not propose this process to construct the network topology in practice. We use this process for analysis and use our peer-pairing algorithm in practice,
which results in random graphs with the same distribution. In the
rest of this section, we analyze the depth of $H^{*}$.
Our analysis to spread over the next two subsections:
\begin{enumerate}
  \item In Section~\ref{subsection:edge_expansion}, we show that it is possible to reach the closest $N/2$ nodes from the source node in $O(\log N)$ hops.
  \item In Section~\ref{subsection:contraction_of_remainging_graphs}, we show that we can reach all the other nodes from the set of the closest $N/2$ nodes in another $O(\log N)$ hops.
\end{enumerate}

\subsection{Edge Expansion of Flow Graph $G_{k}$}
\label{subsection:edge_expansion}
\begin{figure}
\centering \includegraphics[width=0.9\mylength]{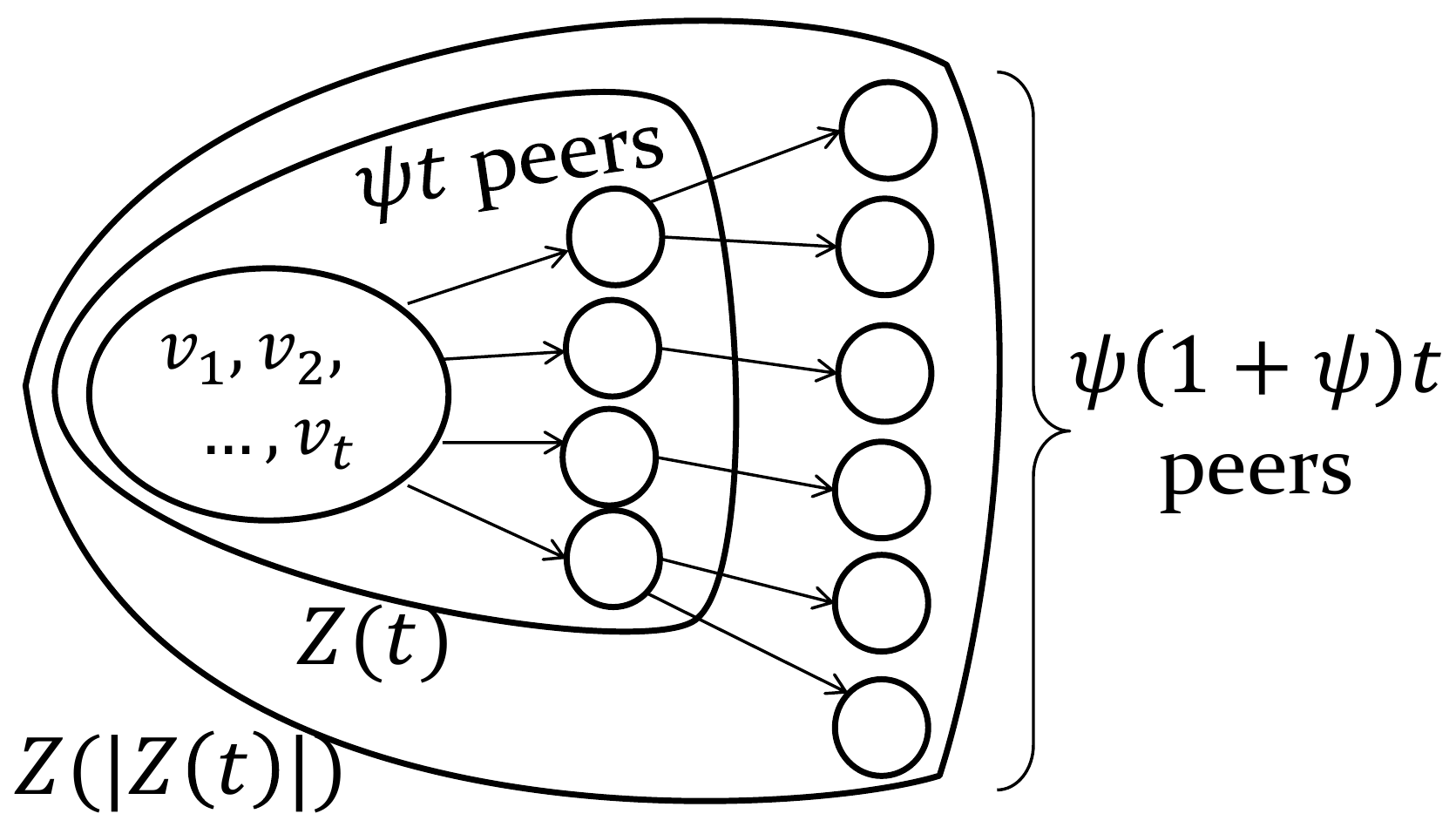}\\
 \caption{If $|Z(t)|>(1+\psi)t$ for some $\psi>0$ in
the flow graph construction process, the number of the peers within $h$ hops from
the first $t$  peers $v_{1},v_{2},\cdots,v_{t}$ is $(1+\psi)^{h}t$, which increases
exponentially in $h$. Thus, $N/2$ peers can be covered within $O(\log_{1+\psi} N)$
steps.}

\label{fig:flow_graph_expansion}
\end{figure}

Before we present our proof that the closest $N/2$ nodes from the source can be reached in $O(\log N)$ hops, we first present some intuition behind the result.
Let $E^{(1)}(t), E^{(2)}(t),$ and $Z(t)$ be $E^{(1)}$, $E^{(2)}$, and $Z$, respectively, at the end of iteration $t$ in the FGC process. By definition, $Z(t)$ is the set of peers $\{v_1, \cdots, v_t\}$ and their children at the end of iteration $t$.
Given these definitions, the proof can be broken into three major steps:\\
\textbf{Step~(i):} We show that $E[z(t)] \geq (1+\psi) t$ for $t\leq N/2$ and some $\psi > 0$, where $z(t)=|Z(t)|$. This is shown in Proposition~\ref{prop:N-z(t)_on_k}.\\
\textbf{Step~(ii):} We show that $z(t) \geq (1+\psi) t$  w.h.p. in Proposition~\ref{prop:expansion_of_flowgraph}.\\
\textbf{Step~(iii):} Finally, we relate the above concentration result to the distance between the source node and its $N/2$ closest peers in Proposition~\ref{prop:P_dG(N/2)}. The intuition behind this result is as follows:  As we can see in Fig.~\ref{fig:flow_graph_expansion}, $Z(t)$ is the set of peers that are within one hop from  the set of the first $t$ peers, i.e., $\{v_1, \cdots, v_t\}$. Similarly, $Z(|Z(t)|)$ is the set of peers that are within one hop from the set of the first $|Z(t)|$ peers, i.e, $\{v_1, \cdots, v_{|Z(t)|}\}$. In general, if we iteratively define $Z^{(h)}(t)=Z( |Z^{(h-1)}(t)|)$ for $h>0$ where $Z^{(0)}(t)=\{v_1, \cdots, v_t\}$, $Z^{(h)}(t)$ is the set of peers that are within one hop from $Z^{(h-1)}(t)$. Thus, the result of \textbf{Step~(ii)} leads to
\begin{align}
z^{(h)}(t)&=z(z^{(h-1)}(t))\geq (1+\psi)z^{(h-1)}(t)\NN\\
&\geq \cdots\geq   (1+\psi)^h t,\label{eq:expanding_half}
\end{align}
where $z^{(h)}=|Z^{(h)}|$. In other words, the graph expands at least at rate $1+\psi$ as shown in Fig.~\ref{fig:flow_graph_expansion}, and thus, it needs $O(\log_{1+\psi} N)$ steps to cover the first $N/2$ peers.

%\begin{figure}[ht]
%  % Requires \usepackage{graphicx}
%  \centering
%        \includegraphics[width=0.95\mylength]{expansion}
%  \caption{Redrawn flow graph $G_1$ in the example of Fig.~\ref{fig:flow_graphs} based on the distance from the source peer: $v_1, v_2, \cdots, v_6$ denote the peer indices sorted by the distance from the source in $G_1$  and $Z(t)$ denotes the set of peers $v_1, v_2, \cdots, v_t$ and their children in  $G_1$.}\label{fig:expansion}
%\end{figure}

Now, we are ready to make the above argument precise. To derive the mean of $Z(t)$ in Step~(i), we focus on how $Z(t)$ increases. At iteration $t$, if peer $v_t$ chooses $c_{t}$ from $Z^C(t-1)$, adding this peer to $Z(t)$ will increase
$z(t)$ by one from $z(t-1)$. In addition, if \tred{$\mu(v_{t})=k$} \tblue{$\tau_t=1$} and peer
$c_{t}'$ is not in $Z(t)=Z(t-1)\cup\{c_{t}\}$, adding peer $c_{t}'$
will also increase $z(t)$ by one. Hence, the increment of $z(t)$
at iteration $t$ is given by
\begin{equation}
z(t)=z(t-1)+\indicate{c_{t}\notin Z(t-1)}+\tred{\indicate{\mu(v_t)=k}} \tblue{\tau_t} \cdot\indicate{c_{t}'\notin Z(t-1)\cup\{c_{t}\}}\label{eq:update_of_zt}
\end{equation}
At each iteration, the increment of $z(t)$ is either 0, 1, or 2, and thus $z(t)-z(0)\leq 2t$. Initially, the increment is 2 with high probability because $Z(t-1)$ contains a
few peers compared with $V$. Therefore, $z(t)$ will increase fast initially.  As $t$ increases, i.e., as $Z(t-1)$ contains more
peers, the probability that the increment is 0 or 1 increases, and
thus $z(t)$ will increase at a slower rate.
Define
$$\cG(t)\triangleq\{E^{(1)}(t), E^{(2)}(t)\}$$
as  the graph drawn right after iteration $t$. Since all outgoing edges of $v_1, \cdots, v_t$ are determined at this moment, $\cG(t)$ determines $Z(t)$ and $z(t)$. The following proposition
shows how the mean of $z(t)$ evolves conditioned of $\cG(l)$ (for any $l\leq t$).
\begin{proposition}\label{prop:N-z(t)_on_k}
For any integer $l$ in $[0, t]$,
\begin{align}
&E\big[N-z(t)\big|\cG(l),\tred{\mu}\tblue{\vec \tau} \big]\NN\\
&=\frac{N-t-1}{N-l-1}\left(\frac{N-\summu{t}-1}{N-\summu{l}-1}\right)(N-z(l)).\label{eq:doob_martingale}
\end{align}
\end{proposition}
The proof is provided in Appendix~\ref{proof:prop:N-z(t)_on_k}.
Since \tred{$\indicate{\mu(v_1)=k}, \indicate{\mu(v_2)=k}, \cdots, \indicate{\mu(v_N)=k}$} \tblue{$\tau_1, \tau_t,\cdots, \tau_N$} are independent Bernoulli random variables with mean $q$, we have $E[\summu{t}]=qt$. By taking $l=0$, we can obtain $E[z(t)]$  from Proposition~\ref{prop:N-z(t)_on_k}:
\begin{align}
 & E\big[N-z(t)\big]=(N-t-1)\left(1-\frac{tq}{N-1}\right)\NN\\
 & \Rightarrow\frac{E[z(t)]}{t}=1+\frac{1}{t}+q(1-\frac{t}{N-1}).\label{eq:mean_of_edge_expansion}
\end{align}
For $t\leq N/2$, we can see that the minimum of $E[z(t)/t]$ is attained
at $t=N/2$, and the minimum is greater than $1+q/2$. Thus, we have
\[
\frac{E[z(t)]}{t}>1+\frac{q}{2},\;\;\;\;\text{ for }t\leq\frac{N}{2}.
\]
In other words, the expected number of outgoing edges from the set of peers $v_{1},v_{2},\cdots,v_{t}$ is at least $\psi t$, which corresponds to Step~(i).

We next show Step~(ii) by showing that $z(t)/t$ is concentrated around its mean \tblue{with high probability}, and thus is larger than $1+q/2$ \tblue{with high probability}. Using a Doob Martingale
and the Azuma-Hoeffding bound, we have the following result: \begin{proposition}\label{prop:expansion_of_flowgraph}
For $\psi\in(0,\frac{q}{2})$ and $t\leq\frac{N}{2}$,
\begin{equation}
\prob{z(t)>(1+\psi)t}>1-\exp\left(-\sigma t\right),\NN
%\label{eq:doob_extension}
\end{equation}
where $q=1/(K-1)$ and $\sigma=(q/2-\psi)^{2}/{8}$.
\end{proposition} The proof is provided in Appendix~\ref{proof:prop:expansion_of_flowgraph}. This result corresponds to the result of Step~(ii).

We finally show how the distance to the closest $N/2$ peers from the source in $O(\log N)$. As we discussed in Step~(iii), we repeatedly apply Proposition~\ref{prop:expansion_of_flowgraph} to (\ref{eq:expanding_half}) for $t=z^{(0)}(t_0),z^{(1)}(t_0),\cdots$ for some $t_0= O(\log N)$ until $z^{(h)}(t_0)\geq N/2$.

Let $d(v)$ be the distance from the source to peer $v$ in the random graph $H^{*}$ constructed by the FGC process. Since $Z(t)$ is the set of the first $t$ peers and their children, $\max_{t<i\leq z(t)}d(v_i)-d(v_{t})\leq1$. From the FGC process, it is easy to see that \tblue{peers are added to $Z$ in an increasing order of their distance from peer~1. Since the order of $v_1,\cdots, v_N$ is determined by the order in which peers are added to $Z$, the distances of $v_1,\cdots, v_N$ must be non-decreasing, i.e.,}
$d(v_{t'})\leq d(v_{t''})$
if $t'<t''$. Hence, we have
\begin{equation}
d(v_{z(t)})-d(v_{t})\leq 1,\;\;\;\text{for }0\leq t\leq N.\label{eq:distance_difference}
\end{equation}
From now on,  we slightly abuse notation so that $v_{t}=v_{\lfloor t\rfloor}$ and $z(t)=z(\lfloor t\rfloor)$  for non-integer $t$.
Now, we show that the distance from the source to peer $v_{\lfloor N/2\rfloor}$ is $\Theta(\log N)$, where $\lfloor y \rfloor$ is the maximum integer not exceeding $y$.

Using (\ref{eq:distance_difference}) and Proposition~\ref{prop:expansion_of_flowgraph},
we can derive the distance from the source to peer $v_{N/2}$ as follows:
\begin{proposition} \label{prop:P_dG(N/2)}
For $\psi\in(0,\frac{q}{2})$,
\begin{equation}
\prob{d(v_{N/2}))<\theta}>1-\frac{e^{\sigma}\log_{1+\psi}\frac{N}{2}}{N^{\sigma}},\NN
%\label{eq:PdG2}
\end{equation}
where $\theta =\log N+\log_{1+\psi}\frac{N}{2}$ and $\sigma=(q/2-\psi)^{2}/{8}$.
\end{proposition}
\begin{proof}
 For simplicity, let $\phi=1+\psi$ for some $\psi\in (0,q/2)$. Define event $A_{t}\triangleq\{z( \phi^{t}\log N)>\phi^{t+1}\log N\}$.  From (\ref{eq:distance_difference}):
\[
 d(v_{z(  \phi^{t}\log N )})-
        d(v_{  \phi^{t}\log N })\leq 1.\NN
\]
If $A_t$ is true, $z(  {\phi^{t}\log N} )> \phi^{t+1}\log N \geq \lfloor \phi^{t+1}\log N \rfloor$. Since $d(v_t)$ is non-decreasing in $t$, the above inequality can be rewritten as
\[
d(v _{ {\phi^{t+1}\log N} })-
        d(v_{ {\phi^{t}\log N} })\leq 1.
\]
(Recall that we have abused notation $v_t$ such that $v_t= v_{\lfloor t \rfloor}$ for non-integer $t$.)
If $A_{t}$ is true for $t=0,1,2,\cdots,t_{0}-1$, we obtain
\begin{align}
 & d( v_{\phi\log N})-d( v_{\log N})\leq1\NN\\
 & d( v_{\phi^{2}\log N})-d( v_{\phi\log N})\leq1\NN\\
 & \;\;\;\;\;\;\;\;\;\vdots\NN\\
 & d( v_{\phi^{t_{0}}\log N})-d( v_{\phi^{t_{0}-1}\log N})\leq1\NN\\
 & \Rightarrow d( v_{\phi^{t_{0}}\log N})-d( v_{\log N})\leq t_{0}\NN\\
 & \Rightarrow d( v_{\phi^{t_{0}}\log N})\leq t_{0}+d(  v_{\log N})\leq t_{0}+\log N.\NN
\end{align}
 Take $t_{0}=\lceil \log_{\phi}\frac{N}{2}-\log_{\phi}\log N\rceil$. For large $N$,
\begin{align}
&d(v_{\frac{N}{2}})\leq  d( v_{\phi^{t_{0}}\log N})\NN\\ <&\log_{\phi}\frac{N}{2}-\log_{\phi}\log N+1+ \log N <\theta.\NN
\end{align}
Using the union bound, the probability of $d(v_{N/2})<\theta$ can
be expressed as
\begin{align}
 &\prob{d(v_{\frac{N}{2}})<\theta}\geq\prob{\cap_{t=0}^{t_{0}-1}A_{t}}\NN\\
 &=1-\prob{\cup_{t=0}^{t_{0}-1}A_{t}^{C}}\geq1-\sum_{t=0}^{t_{0}-1}\prob{A_{t}^{C}}
 \label{eq:patc}
\end{align}

We next find an upper bound on $\prob{A_{t}^{C}}$. By the definition of $A_t$,
\begin{align}
&\prob{A_{t}^{C}} = \prob{z( \phi^{t}\log N) \leq \phi^{t+1}\log N}\NN\\
=&\prob{z(\phi^{t}\log N)\leq
(1+\psi_t) \lfloor \phi^{t}\log N \rfloor},\label{eq:prob_median}
\end{align}
where
$$1+\psi_t= \frac{\phi^{t+1}\log N}{ \lfloor\phi^{t}\log N\rfloor}.$$
Suppose $\phi^{t}\log N= a+b$, where $a$ is an integer and $b\in [0,1)$. Since $\phi<2$, we have $\phi^{t+1}\log N=\phi \cdot a + \phi\cdot  b< \phi\cdot  a+2$. Thus, $1+\psi_t<\phi +2/  (\phi^{t}\log N)= \phi+o(1)$, which is in $(1, 1+q/2)$ for
large $N$. Applying this and Proposition~\ref{prop:expansion_of_flowgraph} to (\ref{eq:prob_median}) for large $N$ and $0\leq t <t_0$,
\begin{align}
&\prob{A_{t}^{C}}< e^{-\sigma_t  \lfloor \phi^t \log N\rfloor}\NN\\
&<e^{-\sigma_{\min} (\phi^t \log N-1)}\NN\\
&\leq e^{-\sigma_{\min}  (\log N-1)},\NN
\end{align}
where $\sigma_t=(q/2-\psi_t)^{2}/8$ and $\sigma_{\min}=\min_{0\leq t <t_0} \sigma_t$.
Since $1+\psi_t< 1+\psi+ o(1)$, we have $\sigma_{\min} \geq \sigma$ for sufficiently large $N$.
 Applying this to (\ref{eq:patc}), we have
\begin{align}
 & \prob{d(v_{\frac{N}{2}})<\theta}\NN\\
  >& 1-\sum_{t=0}^{t_{0}-1}e^{-\sigma \log N}e^{\sigma}\NN\\
 \geq &1-t_0 e^{-\sigma  \log N}e^{\sigma}\NN\\
  >&1-\frac{e^{\sigma}\log_{\phi}\frac{N}{2}}{N^{\sigma}},\NN
\end{align}
for large $N$.
\end{proof}
Since $d(v_t)\leq d(v_{N/2})$ if $t<N/2$, this proposition shows that the closest $N/2$ peers are within $O(\log N)$ hops from the source, which corresponds to the result of Step~(iii).

In this subsection, we have found that the distance from the source to $v_{N/2}$ is $\Theta(\log N)$ by analyzing  edge expansion in the early phase $t\leq N/2$. To show the distance from the source to the farthest peer $v_N$, we will analyze the distance from $v_{N/2}$ to $v_N$.

\subsection{Contraction of the Remaining Graph}
\label{subsection:contraction_of_remainging_graphs}

To analyze the distance of the remaining nodes from the source, we need to show that the distance from closest $N/2$ nodes from the source to the remaining $N/2$ nodes is $O(\log N)$. However, we cannot use the same approach as in the previous subsection because the edge-expansion analysis using  Proposition~\ref{prop:expansion_of_flowgraph}  is not valid for $t>\frac{N}{2}$. Indeed, we can infer from (\ref{eq:mean_of_edge_expansion}) that $E[z(t)/t]$ reduces to one for large $N$, which indicates almost zero expansion.

Instead of edge expansion, we focus of the contraction on the number of remaining peers. Recall that $Z^{(h)}(N/2)$ is the set of peers that are within $h$ hops from the closest $N/2$ peers (i.e., $\{v_1, \cdots, v_{N/2}\}$), and $z^{(h)}(N/2)$ is the number of such peers. Hence, $N-z^{(h)}(t_0)$ is the number of peers that are $h+1$ or more hops away from the closest $N/2$ peers. Using this notation, we show that  $O(\log N)$ peers  are not within $O(\log N)$ hops from the closest $N/2$ peers, i.e., $N-z^{( \Theta(\log N) )}(N/2)=O(\log N)$. We deal with the final $O(\log N)$ peers separately at the end.

Before we prove this, we present the intuition behind the proof:\\
\textbf{Step~(i):} To observe the contraction of the number of remaining peers at each iteration $t$, we define \emph{the contraction ratio} at iteration $t$ as
$$F(t)\triangleq \frac{N-z(t)}{N-t}.$$
Since a small $F(t)$ means a large number of peers are within one hop from $\{v_1,\cdots, v_t\}$, we want $F(t)$ to be small for a faster contraction.  By proving that $F(t)$ is a supermartingale, we first show that
the mean of the contraction ratio at any iteration $t'>t$ is no larger than that at iteration $t$, i.e., $E[F(t)]\geq E[F(t')]$,  in Lemma~\ref{lemma:super_martingale_F}.  Further,
 we show that the contraction ratio at iteration $t'$ is no larger than that at iteration $t$ with high probability, i.e., $F(t)\geq F(t')$ for  $t<t'$ in Proposition~\ref{prop:shrinking_martingale_concentration}. This implies that the contraction at iteration $t'$ is no smaller than that at iteration $t$.  From this, we can conclude that, if we achieve a small contraction ratio at $t=N/2$, then we will have a small contraction ratio afterwards.

\textbf{Step~(ii):}  In this step, we show that the result in Step~(i)  holds with high probability for all iterations $t_0<t_1<t_2< \cdots < t_{D}$ that satisfy  $N-t_h=F(t_0)^h (N-t_0)$  for $0<h \leq D$,  where $D$ is an appropriately chosen number which is of the order of $\log N$, i.e.,  the contraction ratio at each iteration $t_h$ is upper bounded by $F(t_0)$, i.e.,
\begin{equation}
F(t_0)\geq F(t_h),\;\; \forall 0<h\leq D ,\label{eq:continuing_bound}
\end{equation}
with high probability. Next, it is proven that, if (\ref{eq:continuing_bound}) is true, we have $z^{(D)}(t_0) \geq t_{D}$. By the definition of $t_h$, we conclude that
$$N-z^{(D)}(t_0) \leq F(t_0)^D (N-t_0)$$
with high probability. This means that, if $t_0=\lfloor N/2 \rfloor$,  the number of peers that are not within $D$ hops from the first $N/2$ peers contracts exponentially in $D$ if the initial contraction ratio $F(t_0)$ is upper bounded by some constant in $(0,1)$. In other words, this means that almost all peers are within $O(\log N)$ hops from the first $N/2$ peers with high probability. The detailed proof of Step~(ii) is in Proposition~\ref{prop:shinking_portion}.

\textbf{Step~(iii):} Finally, we show that  $F(\lfloor N/2 \rfloor)$ is upper bounded by some constant in $(0,1)$ with high probability in Lemma~\ref{lemma:beginning_factor}. With this bound on $F(\lfloor N/2 \rfloor)$ and the exponential contraction in Step~(ii), we show that  all peers except $\log N$ peers must be within $\Theta(\log N)$ hops from the first $N/2$ peers with high probability. We then  show that the maximum distance to  the remaining peers is also $O(\log N)$. Overall, the maximum distance to all peers from the first $N/2$ peers is $\Theta(\log N)$. We show this in Proposition~\ref{prop:final_result}.

For Step~(i), we  need to compare two contraction ratios $F(t_0)$ and $F(t)$ for $t>t_0$.  As an initial step, we show that the mean $E[F(t)]$ of the contraction ratio does not increase. Recall that $\cG(t)$ is the graph drawn up to the end of iteration $t$. Since $\cG(t)$ determines $z(t)$, it also determines $F(t)$. The next lemma shows that $F(t)$ is a supermartingale conditioned on $\cG(t)$.
\begin{lemma}\label{lemma:super_martingale_F}
$F(t)$ is a supermartingale, i.e., for $t>0$,
$$E[\sfactor(t)\;|\;\cG(t-1)]\leq\sfactor(t-1).$$
\end{lemma}
\begin{proof}
Note that if $l \leq t$, we have ${N-\summu{t}-1}\leq {N-\summu{l}-1}$  since $\summu{l}$ is the number of edges in $E^{(2)}$ up to iteration $l$. Thus, from Proposition~\ref{prop:N-z(t)_on_k},
we have
\begin{align}
&E[N-z(t)|\cG(t-1),\tblue{\vec \tau}\tred{\mu}]\NN\\
&\leq\frac{N-t-1}{N-(t-1)-1}(N-z(t-1))\NN\\
&<\frac{N-t}{N-(t-1)}(N-z(t-1)).\NN
\end{align}
If we divide both sides by $N-t$, we have $E[F(t)|\cG(t-1), \tblue{\vec \tau}\tred{\mu}]<F(t-1)$. Since the upper bound is independent of \tred{coloring decision $\mu$} \tblue{$\vec \tau$}, we have the result in the lemma.
\end{proof}

The lemma implies that for any $t>t_0$, the expected contraction ratio $E[F(t)]\leq E[ F(t_0)].$   Using the Azuma-Hoeffding inequality, we can prove that  $F(t_0)\geq F(t)$ with high probability.
\begin{proposition}\label{prop:shrinking_martingale_concentration}
For every $\epsilon>0$ and $0<t_0<t$, we have
$$\prob{F(t)-F(t_0)>\epsilon|\cG(t_0)}<\exp\left(-\frac{\epsilon^2 (N-t)}{8}\right).$$
\end{proposition}
\begin{proof}
We use the Azuma-Hoeffding inequality for supermartingale $F(t)$. To this end, we first need to find an upper bound of $|F(t)-F(t-1)|$. Note that
\begin{equation}
F(t)-F(t-1)=\frac{N-z(t)}{N-t}-\frac{N-z(t-1)}{N-t+1}.\label{eq:lipschitz_eq1}
\end{equation}
Since $z(t)-z(t-1)\leq 2$ by (\ref{eq:update_of_zt}), (\ref{eq:lipschitz_eq1}) is lower bounded by
$$\frac{  N-z(t)-(N-z(t-1))}{N-t+1}\geq -\frac{2}{N-t+1}.$$
Since $z(t)\geq z(t-1)\geq t-1$, (\ref{eq:lipschitz_eq1}) is upper bounded by
$$\frac{ N-z(t-1)}{(N-t)(N-t+1)}\leq \frac{2}{N-t+1}.$$
Hence, $|F(t)-F(t-1)|$ must be upper bounded by $2/(N-t+1)$.

 We now apply this Lipschitz difference to the Azuma-Hoeffding bound:
\begin{align}
&\prob{F(t)-F(t_0)\geq \epsilon|\cG(t_0)}\NN\\
<&\exp\left(-\frac{\epsilon^2}{2\sum_{l=t_0+1}^{t}\frac{4}{(N-l+1)^2}}\right)\NN\\
=&\exp\left(-\frac{\epsilon^2}{8\sum_{l=N-t+1}^{N-t_0}\frac{1}{l^2}}\right).\label{eq:azuma_upper}
\end{align}
For $0< a <b$, we have
\begin{align}
\sum_{l=a+1}^b \frac{1}{l^2}\leq \int_{a}^b \frac{1}{x^2}dx\leq \frac{1}{a}.\NN
\end{align}
Applying the above for $a=N-t$ and $b=N-t_0$, the R.H.S. of (\ref{eq:azuma_upper}) is upper bounded by
$\exp( -\epsilon^2(N-t)/8)$.
\end{proof}
This proposition implies that the contraction ratio at a given iteration $t_0$ will be an upper bound of that at a later iteration, which corresponds to Step~(i)

%\begin{lemma}\label{lemma:shrink_continues}
%Let  $\phi=F(t_0)+\epsilon$ for some $t_0>0$ and $\epsilon\in(0, 1-F(t_0))$.
%Define $t_k$ such that $N-t_k=\lfloor\phi^k (N-t_0)\rfloor$ for $k>0$. If $F(t_k) \leq \phi$ for all $0\leq k\leq D$,
%we have
%$$ t_{D+1}\leq z^{(D+1)}(t_0).$$
%\end{lemma}
%\begin{proof}
%We show by induction that if $A_0,\cdots, A_k$ are true, we have $z^{(k+1)}(t_0)>t_{k+1}$. For $k=0$, since $N-z^{(1)}$ is an integer, we have
%\begin{align}
%&F(t_0)=\frac{N-z(t_0)}{N-t_0}\leq \phi\NN\\
%\Rightarrow & N-z(t_0)\leq \phi(N-t_0)\NN\\
%\Rightarrow &N-z(t_0)\leq \lfloor \phi(N-t_0) \rfloor =N-t_1.\NN
%\end{align}
%Thus, we have $ t_1\leq z(t_0)=z^{(1)}(t_0)$.
%
%We assume that if $A_0,\cdots, A_{k-1}$ are true, we have $z^{(k)}(t_0)>t_{k}$. We then show that if $A_0,\cdots, A_{k}$ are true, we have $z^{(k+1)}(t_0)>t_{k+1}$.
%If  $A_k$ is true,  i.e., $F(t_k)\leq \phi$, we have
%\begin{align}
%N-z(t_k)&\leq \phi(N-t_k)= \phi \lfloor \phi^k (N-t_0)\rfloor\NN\\
%\Rightarrow N-z(t_k)&\leq  \lfloor\phi\lfloor \phi^k (N-t_0)\rfloor\rfloor\NN\\
%&\leq \lfloor \phi^{k+1} (N-t_0)\rfloor=N-t_{k+1}.\label{eq:shirnk_continue1}
%\end{align}
%By the induction hypothesis, we have $t_k\leq z^{(k)}(t_0)$. Since $z(t)$ is non-decreasing,  we have $z(t_k)\leq z(z^{(k)}(t_0))=z^{(k+1)}(t_0)$. Applying this to (\ref{eq:shirnk_continue1}),  we have $t_{k+1}\leq z^{(k+1)}(t_0)$. By induction,  we conclude the result of the lemma.
%\end{proof}

In Step~(ii), we show that the exponential contraction of the remaining peers holds with high probability over multiple iterations using the result in Step~(i).  After establishing this, we show in the next proposition that all, but $\log N$, peers are within $O(\log N)$ hops from the closest $N/2$ peers from the source.
\begin{proposition}\label{prop:shinking_portion}
Fix $t_0=\left\lfloor \frac{N}{2}  \right\rfloor$.
Conditioned on $\cG(t_0)$, let $\phi=F(t_0)+\epsilon$ for an arbitrary $\epsilon\in (0, 1-F(t_0))$. Then, for large $N$,
\begin{align}
 &\prob{ N-z^{(D(\phi))}(t_0) \leq \log N\;\Big|\; \cG(t_0)}\NN\\
 &\;\;\;\;\;\;\;\;>1-D(\phi)\cdot  N^{-\epsilon^2/8},\label{eq:prob_lowerbound_main}
\end{align}
where
$D(\phi)=\left\lfloor \log_{1/\phi} (  (N-t_0) /\log N)  \right\rfloor.$
\end{proposition}
\begin{proof}
Define $t_h$ such that $N-t_h=\lfloor\phi^h (N-t_0)\rfloor$ for $h>0$.
 Define event $A_h=\{F(t_h)\leq\phi\}$ for $h\geq 0$. Since $F(t_0)=\phi-\epsilon$, $A_0$ is always true. For convenience, let $D=D(\phi)$. From Proposition~\ref{prop:shrinking_martingale_concentration}, the probability that $A_0,A_1, \cdots, A_{D-1}$ are all true is lower bounded as
\begin{align}
&\prob{\cap_{h=0}^{D-1} A_h\;|\;\cG(t_0)}\NN\\
&\geq 1-\prob{A_0^C|\cG(t_0)}-\sum_{h=1}^{D-1} \prob{A_h^C|\cG(t_0)}\NN\\
&\geq 1-\sum_{h=1}^{D-1} \exp\left( -\frac{\epsilon^2 (N-t_h)}{8}\right)\NN\\
&\geq 1-(D-1)\exp\left(  -\frac{\epsilon^2 (N-t_{D-1})}{8} \right).\label{eq:prob_lowerbound}
\end{align}
For large $N$, we have
\begin{align}
N-t_{D-1} &= \lfloor \phi^{D-1} (N-t_0) \rfloor > \phi^D(N-t_0)\NN\\
&\geq \log N.\NN
\end{align}
Hence, the R.H.S. of (\ref{eq:prob_lowerbound}) is lower bounded by
$1-(D-1)\exp(-\epsilon^2\log N/8)>1-D\cdot N^{-\epsilon^2/8}$, which is equal to the lower bound in (\ref{eq:prob_lowerbound_main}).

We only need to show that the probability in (\ref{eq:prob_lowerbound}) is upper bounded by that in (\ref{eq:prob_lowerbound_main}). To prove this, we  show by induction that, if $A_0,\cdots, A_h$ are true, $z^{(h+1)}(t_0)>t_{h+1}$ is also true. For $h=0$, since $N-z^{(1)}(t_0)$ is an integer, we have
\begin{align}
&F(t_0)=\frac{N-z(t_0)}{N-t_0}\leq \phi\NN\\
\Rightarrow & N-z(t_0)\leq \phi(N-t_0)\NN\\
\Rightarrow &N-z(t_0)\leq \lfloor \phi(N-t_0) \rfloor =N-t_1.\NN
\end{align}
Thus, we have $ t_1\leq z(t_0)=z^{(1)}(t_0)$.

We assume that, if $A_0,\cdots, A_{h-1}$ are true,  $t_{h}\leq z^{(h)}(t_0)$ is also true. We then show that if $A_0,\cdots, A_{h}$ are true,  $z^{(h+1)}(t_0)>t_{h+1}$ is also true.
If $A_h$ is true,  i.e., $F(t_h)\leq \phi$, we have
\begin{align}
N-z(t_h)&\leq \phi(N-t_h)= \phi \lfloor \phi^h (N-t_0)\rfloor\NN\\
\Rightarrow N-z(t_h)&\leq  \lfloor\phi\lfloor \phi^h (N-t_0)\rfloor\rfloor\NN\\
&\leq \lfloor \phi^{h+1} (N-t_0)\rfloor=N-t_{h+1}.\label{eq:shirnk_continue1}
\end{align}
By the induction hypothesis, we have $t_h\leq z^{(h)}(t_0)$. Since $z(t)$ is non-decreasing,  we have $z(t_h)\leq z(z^{(h)}(t_0))=z^{(h+1)}(t_0)$. Applying this to (\ref{eq:shirnk_continue1}),  we have $t_{h+1}\leq z^{(h+1)}(t_0)$.

By induction,  we conclude that, if $A_0, \cdots, A_{D-1}$ are true, we have $t_{D}\leq z^{(D)}(t_0)$, which is equivalent to $N-z^{(D)}(t_0)\leq N-t_{D}$.
By definition of $t_{D}$, we have
$N-t_{D} \leq  \phi^{D} (N-t_0) < \log N.$ Thus,  if $A_0, \cdots, A_{D-1}$ are true, we have $N-z^{(D)}(t_0)\leq \log N$. Hence, the probability in (\ref{eq:prob_lowerbound}) is upper bounded by that in (\ref{eq:prob_lowerbound_main}), which leads to the result of this proposition.
\end{proof}

From Proposition~\ref{prop:shinking_portion}, if $\phi$ is upper bounded by a constant less than one, it follows that $D(\phi)=\Theta(\log N)$. This means that all peers except at most $\log N$ peers are within $O(\log N)$ hops from the first $N/2$ peers with high probability. To show that all peers are also within $O(\log N)$ hops from the first $N/2$ peers, we only need to show that there exists the upper  bound on  $F(N/2)$ and the distance to the remaining $\log N$ peers from the other peers is also $O(\log N)$. We show these in Step~(iii).

In Step~(iii), we first show that $F(N/2)$ is upper bounded:
\begin{lemma}\label{lemma:beginning_factor}
For any $\epsilon\in (0, q/2)$, we have
\begin{equation}
\prob{F(\left\lfloor \frac{N}{2} \right\rfloor)\geq 1-\frac{q}{2}+\epsilon}\leq  \exp\left(-\frac{\epsilon^2}{32} \left\lfloor \frac{N}{2} \right\rfloor \right),\NN
%\label{eq:F0}
\end{equation}
for sufficiently large $N$.
\end{lemma}
\begin{proof}
Let $t=\lfloor N/2 \rfloor$, $\epsilon'=\epsilon/2$,  and $\psi= q/2-\epsilon'$, where we recall that $q=1/(K-1)$. We can rewrite Proposition~\ref{prop:expansion_of_flowgraph} as follows:
$$\prob{z(t)\leq (1+\frac{q}{2}-\epsilon')t}\leq  \exp\left(-\frac{\epsilon'^2}{8} \left \lfloor\frac{N}{2}  \right\rfloor \right).$$
Note that
\begin{align}
&z(t )\leq(1+\frac{q}{2}-\epsilon')t\NN\\
&\Leftrightarrow \frac{N-z(t )}{N- t } \geq \frac{\frac{N}{t}-1-\frac{q}{2}+\epsilon'}{\frac{N}{t}-1}.\NN
\end{align}
Since $N/t-1\geq 1$ and $t\geq (N-1)/2$, we have
\begin{align}
&\frac{\frac{N}{t}-1-\frac{q}{2}+\epsilon'}{\frac{N}{t}-1}
<\frac{N}{t}-1-\frac{q}{2}+\epsilon'\NN\\
\leq& 2+\frac{2}{N-1}-1-\frac{q}{2}+\epsilon'<1-\frac{q}{2}+2\epsilon',\NN
\end{align}
if $\epsilon'>2/(N-1)$, which is true for sufficiently large $N$.
Hence, we have  $\prob{F(t)\geq 1-\frac{q}{2}+2\epsilon'}\leq \prob{z(t)\leq (1+q/2-\epsilon')t}$. Since $\epsilon=2\epsilon'$, we have the result of this lemma.
%By (\ref{eq:update_of_zt}), $z(N/2)\leq N/2+\summu{N/2}$. Thus,
%\begin{align}
%&\prob{\sfactor{0}\leq 1-q-\epsilon}\NN\\
%=&\prob{N-z(\frac{N}{2})\leq\frac{N}{2}(1-q-\epsilon)}\NN\\
%=&\prob{z(\frac{N}{2})\geq \frac{N}{2}(1+q+\epsilon)}\NN\\
%\leq&\prob{\summu{\frac{N}{2}}\geq \frac{N}{2}(q+\epsilon)}.\NN
%\end{align}
%Since $\summu{\frac{N}{2}}$ is the sum of $N/2$ independent Bernoulli random variables with mean $q$, we apply the above to the Hoeffding bound:
%\begin{align}
%\prob{\summu{\frac{N}{2}}\geq \frac{N}{2}(q+\epsilon)}\leq \exp\left( -\epsilon^2N \right).\NN
%\end{align}
%Using the union bound, we have
%\begin{align}
%&\prob{1-q-\epsilon <\sfactor{0}< 1-q/2+\epsilon}\NN\\
%\geq &1-\prob{\sfactor{0}\geq 1-q/2+\epsilon}-\prob{1-q-\epsilon \leq \sfactor{0}}\NN\\
%> &1-2\exp\left(  -\frac{\epsilon^2 N}{16}\right).\NN
%\end{align}
\end{proof}

Finally, we prove that all peers are within $\Theta(\log N)$ hops from the first $N/2$ peers using all the results in this subsection.
\begin{proposition}\label{prop:final_result}
For any $\psi\in(0, q/2)$,
\begin{align}
&\prob{d(v_N)-d(v_{N/2})\leq \theta}\NN\\
\geq& 1-\frac{\log_{1+\psi}\frac{N}{2}}{N^{\sigma'}}-\exp(-\frac{\sigma'}{4} N),\NN
\end{align}
where $\sigma'=(q/2-\psi)^2/32$.
\end{proposition}
\begin{proof}
Let $D^*=D(1-q/2+\epsilon)$ and $t_0=\lfloor N/2 \rfloor$. To show this proposition, we first show the following inequality:
\begin{align}
&\prob{d(v_N)-d(v_{t_0})\leq \log N+D^*}\NN\\
&\geq P[d(v_N)-d(v_{N-\log N})\leq \log N,\NN\\
&\;\;\;\;\; \;\;\;\;\;     d(v_{N-\log N})-d(v_{t_0})\leq D^*]\NN\\
&=\prob{d(v_{N-\log N})-d(v_{t_0})\leq D^*}\label{eq:final_equi_prob1}\\
&\geq\prob{z^{(D^*)}(t_0)\geq N-\log N}.\label{eq:final_equi_prob2}
\end{align}
We have obtained (\ref{eq:final_equi_prob1}) from the fact that $d(v_N)-d(v_{N-\log N})\leq \log N$ is always true because   $d(v_{t+1})-d(v_t)\leq d(v_{z(t)})-d(v_t)\leq 1$ by (\ref{eq:distance_difference}). Recall that $Z^{(D^*)}(t_0)$ is the set of peers that are within $D^*$ hops from the first $t_0$ peers, i.e., $\{v_1,v_2,\cdots, v_{t_0}\}$. Thus, if $v_{N-\log N}\in Z^{(D^*)}(t_0)$ (i.e., $z^{(D^*)}(t_0)\geq N-\log N$), then peer $v_{N-\log N}$ must be within $D^*$ hops from the first $t_0$ peers, i.e., $d(v_{N-\log N})-d(v_{t_0})\leq D^*$. Hence, (\ref{eq:final_equi_prob2}) follows from (\ref{eq:final_equi_prob1}).

Let $\textbf{G}$ be the set of all possible $\cG(t_0)$'s that satisfy $F(t_0)\leq 1-q/2+\epsilon/2$. (Recall that $\cG(t_0)$ determines $F(t_0)$.)
Then,  (\ref{eq:final_equi_prob2}) is lower bounded by
\begin{align}
&\prob{z^{(D^*)}(t_0)\geq N-\log N, \cG(t_0)\in \textbf{G}}\NN\\
=& \prob{z^{(D^*)}(t_0)\geq N-\log N | \cG(t_0)\in \textbf{G}}\label{eq:final_equi_prob3}\\
&\;\;\;\;\;\times\prob{\cG(t_0)\in \textbf{G}}.\label{eq:final_equi_prob4}
\end{align}
Note that $D(\phi)$ is non-decreasing. If $\cG(t_0)\in \textbf{G}$, then $F(t_0)\leq 1-q/2+\epsilon/2$, and thus
$D(F(t_0)+\epsilon/2)\leq D(1-q/2+\epsilon)= D^*$. Using Proposition~\ref{prop:shinking_portion}, we can find that   (\ref{eq:final_equi_prob3}) is lowered bounded by
\begin{align}
&\prob{z^{(D^*)}(t_0)\geq N-\log N | \cG(t_0)\in \textbf{G}}\NN\\
\geq&\prob{z^{(D(F(t_0)+\frac{\epsilon}{2}))}(t_0)\geq N-\log N | \cG(t_0)\in \textbf{G}}\NN\\
\geq&1- \max_{F(t_0): \cG(t_0)\in \textbf{G}} D(F(t_0)+\frac{\epsilon}{2})\cdot  N^{-\epsilon^2/32}\NN\\
\geq&1-D^*\cdot  N^{-\epsilon^2/32}.\label{eq:final_equi_prob5}
\end{align}

From Lemma~\ref{lemma:beginning_factor}, we have found that (\ref{eq:final_equi_prob4}) is lower bounded by $1-\exp(-\epsilon^2 t_0/ 128 )$. Applying  (\ref{eq:final_equi_prob5}) and this to (\ref{eq:final_equi_prob3}) and (\ref{eq:final_equi_prob4}), respectively, we  have
\begin{align}
&\prob{d(v_N)-d(v_{t_0})\leq \log N+D^*}\NN\\
\geq& 1-D^*\cdot  N^{-\epsilon^2/32}-\exp(-\frac{\epsilon^2}{128} t_0).\NN
\end{align}
Since $D^*\leq \lfloor  \log_{\frac{1}{1-q/2+\epsilon}}N/2 \rfloor\leq  \lfloor  \log_{{1+q/2-\epsilon}}N/2 \rfloor$, we have $\log N+D^*\leq \theta$, where $\theta$ was defined in Proposition~\ref{prop:P_dG(N/2)}.
Using this and $\psi=q/2-\epsilon$ to the above, we finally have the result of the proposition.
\end{proof}
This proposition shows that the maximum distance from the first $N/2$ peers to all peers is $O(\log_{1+\phi}) N$ with high probability for some $\phi\in(0,q/2)$.

We  can prove the main theorem of this paper, Proposition~\ref{prop:depth_of_flowgraph}. In the previous subsection, we have shown that
the maximum distance from $v_1$ to the first $N/2$ peers is $O(\log_{1+\psi}N)$ with probability $1-O(\log_{1+\psi}N/N^{\sigma})$. In this subsection, we have shown that the maximum distance from the first $N/2$ peers to all other peers is also $O(\log_{1+\psi}N)$ with probability $1-O(\log_{1+\psi}N/N^{\sigma'})$. Combining both results using the union bound, we can conclude that the maximum distance from the source peer to all other peers is $O(\log_{1+\psi}N)$ with probability $1-O(\log_{1+\psi}N/N^{\sigma'})$.

\section{Diameters of Flow Graphs}\label{sec:diameter}

We have shown that the maximum distance from the source to all peers in a flow graph is $O(\log N)$ with high probability. Using this result, we show that the diameter of the flow graph is also $O(\log N)$ with high probability, i.e., the distance between any pair of peers in a flow graph is $O(\log N)$ with high probability.

 To analyze the diameter, we consider a flow graph
with reversed edges. Specifically, for a given multi-digraph $G$, we reverse the direction of each edge and denote the resulting graph by $\pi(G)$. By definition, the distance
from the source to peer $v$ in  $H^*$ is the same as the distance
from peer $v$ to the source in $\pi({H}^{*})$. Thus, the maximum distance from the source to all peers in $H^*$ is the same as the maximum distance  from all peers to the source in $\pi( H^*)$.
\begin{lemma}\label{lemma:reversed_flow_graph}
Let $d^*$ be the maximum distance from the source to all peers  in $H^*$, and let $\tilde d^*$ be the maximum distance from all peers to  the source  in the same graph $H^*$. Then, $ d^*$ is identically distributed as  $\tilde d^*$, i.e., for all $d\geq 0$
$$\prob{d^*\leq d}=\prob{\tilde d^*\leq d}.$$
\end{lemma}
\begin{proof}
We first need to show that
$H^{*}$ and $\pi({H}^{*})$ have the same distribution. We will then show that this result lead to the same distribution of $d^*$ and $\tilde d^*$.
  Recall that $H^*$ is the superposed graph of $H_1$ and $H_2'$, and thus $\pi(H^*)$ is the superposed graph of $\pi(H_1)$ and $\pi(H_2')$. It is easy to see that $H_1$, $H_2$, $\pi(H_1)$ and $\pi(H_2)$ have the same distribution, i.e., for any Hamiltonian cycle $G$ and $l\in\{1,2\}$,
\[
\prob{{H}_{l}=G}=\prob{\pi({H}_{l})=G}=\frac{1}{(N-1)!}.
\]

For a Hamiltonian cycle $G=(V,E)$,  fix a subgraph $G'=(V,E')$, where $E'$ is a subset of $E$. Recall that we have constructed $H_2'$ by independently removing each edge with probability $1-q$ from $H_2$. Thus, we have
\begin{align}
&\prob{H_2'= G'\;|\; H_2=G}=q^{|E'|}(1-q)^{N-|E'|}\NN\\
=&\prob{H_2'=\pi(G')\;|\;H_2=\pi(G)}\NN\\
=&\prob{\pi(H_2')=G' \;|\; \pi(H_2)=G}.\NN
\end{align}
Since $H_2$ and $\pi(H_2)$ have the same distribution, $H_2'$ and $\pi(H_2')$ must have the same distribution. Therefore, ${H}^{*}$ and $\pi(H^*)$ have the same distribution.

Using this result, we show that $d^*$ and $\tilde d^*$ also have the same distribution. Let $\mathcal{G}$ be the set of all possible flow graphs that satisfy $d^*\leq d$. Then, it is easy to see that $\pi\mathcal{G}$ is the set of all possible flow graphs that satisfy $\tilde d^*\leq d$, where  $\pi\mathcal{G}=\{\pi(G)|G\in\mathcal{G}\}$. Since $H^*$ and $\pi(H^*)$ have the same distribution, we have
\begin{align}
&\prob{H^*\in\mathcal{G}}=\prob{\pi(H^*)\in \mathcal{G}}\NN\\
\Leftrightarrow&\prob{H^*\in\mathcal{G}}=\prob{H^*\in\pi\mathcal{G}}\NN\\
\Leftrightarrow&\prob{d^*\leq d}=\prob{\tilde d^*\leq d}\NN.
\end{align}
Hence, $d^*$ and $\tilde d^*$ have the same distribution.
\end{proof}

Previously, we have shown that the maximum distance from the source to other peers in a flow graph is $\Theta(\log N)$ with high probability. Although this result is enough to show $\Theta(\log N)$  streaming delay, we can prove the following stronger result which was mentioned in the second example of Section~\ref{sec:examples}:
\begin{proposition}\label{prop:diameter_of_flowgraph}
For any $\psi\in (0, q/2)$, the diameter of flow graph $G_{k}$ is $O(\log_{1+\psi}N)$ with probability $1-O(\log_{1+\psi}N/N^{\sigma'})$ for some positive constant $\sigma'$.
\end{proposition}
\begin{proof}
Let $d_{i,j}$ be the minimum distance from peer $i$ to peer $j$ in $H^*$. We show that $\max_{(i,j)} d_{i,j}$ is $O(\log_{1+\psi} N)$ with probability $1-O(\log_{1+\psi}N/N^{\sigma})$. Since $d_{i,j}$ is the minimum distance from peer $i$ to peer $j$,  the length of the shortest path from $i$ to $j$ via the source (peer 1) is upper bounded by $d_{i,j}$. Thus, for any $d>0$,
 \begin{align}
& \prob{\max_{(i,j)\in V^2} d_{i,j}\leq 2d}\NN\\
\geq& \prob{\max_{(i,j)\in V^2} d_{i,1}+d_{1,j}\leq 2d}\NN\\
\geq& \prob{\max_{i\in V} d_{i,1}+ \max_{j\in V} d_{1,j}\leq 2d}\NN\\
\geq& 1-\prob{d^*>d}-\prob{\tilde d^*> d}\NN\\
=& 1-2\prob{d^*>d}.\NN
 \end{align}
In the last equation, we have used Lemma~\ref{lemma:reversed_flow_graph}. In the previous section, we have shown that if $d=\Theta(\log_{1+\psi} N)$,
 $$\prob{d^*>d}<O(\log_{1+\psi}N/N^{\sigma}).$$
Thus, we have  proven the proposition.
%
%Let $H^*$ denote a flow graph, and  Let $\tilde H^*$ be
%a graph with the same set of edges as $H^*$ but with the direction of each edge reversed, and let $\tilde d_{i,j}$ be the distance from $i$ to $j$ in $H^*$. Since the shortest path from $i$ to $j$ in $H^*$ is the shortest path from $j$ to $i$ in $\tilde H^*$, it follows that $d_{i,j}=\tilde d_{j,i}$ for all pairs $(i,j)$. From (\ref{eq:depth_of_k}),
%$$\prob{\max_{j\in V} \tilde d_{j,s}\geq 2\theta }=\prob{\max_{j\in V} d_{s,j} \geq  2\theta}=o(1),$$
%where  $\theta=\log N+\log_{1+\psi}\frac{N}{2}$.
% As we proved in Lemma~\ref{lemma:reversed_flow_graph}, the random graphs $H^*$ and $\tilde H^*$ have the same distribution. Thus, the corresponding $d_{i,j}$ and $\tilde d_{i,j}$ also have the same distribution, i.e.,
%$$\prob{\max_{j\in V} d_{j,s}\geq 2\theta}=\prob{\max_{j\in V} \tilde d_{j,s}\geq 2\theta}=o(1).$$
%Using the union bound, we have
%\begin{align}
%&\prob{\max_{j\in V} d_{j,s}<  2\theta \text{ and } \max_{j\in V} d_{s,j}< 2\theta}\NN\\
%=&1-\prob{\max_{j\in V} d_{j,s}\geq 2\theta \text{ or } \max_{j\in V} d_{s,j}\geq 2\theta}\NN\\
%\geq&1-\prob{\max_{j\in V} d_{j,s}\geq 2\theta}-\prob{\max_{j\in V} d_{s,j}\geq 2\theta}\NN\\
%=&1-2\cdot o(1)=1-o(1)\NN
%\end{align}
%This result implies that for any two peers $i$ and $j$, both the distance from $i$ to the source and the distance from the source to $j$ are $O(\log N)$ with high probability. Thus, the distance from $i$ to $j$ must be $O(\log N)$ with high probability.
\end{proof}

\section{Conclusions}

\label{sec:Conclusions} Instead of conventional approaches using
multiple overlay trees, we have proposed a simple P2P streaming algorithm   that consists of a simple pairing algorithm similar to the one proposed earlier for constructing distributed hash tables \cite{Law03}, but used here for streaming data. Our proposed  chunk dissemination algorithm can deliver all chunks to all peers with $\Theta(\log N)$ delay and achieves $(1-1/K)$ fraction  of the optimal streaming capacity for any constant $K\geq2$.

There are several issues that need to be addressed to implement our algorithm in practice. The first issue is one of modifying our chunk dissemination algorithm to accommodate peer churn. Even though our chunk dissemination algorithm shows that the network topology at any given moment can achieve both a near optimal throughput and $\Theta(\log N)$ delay,
 peer arrivals and departures disrupt the topology continuously and hence, one needs practical solutions to account for this churn in the chunk dissemination algorithm as well as the delay analysis.
 The second issue is one of dealing with packet losses. Even in the wired Internet, packet losses are not uncommon, and therefore, a practical protocol must have provisions to recover from such losses. Finally, we have to deal with asynchronous transmissions, i.e., chunk transmissions will not occur in a time-slotted manner in the Internet for many reasons. The analysis gets much more involved in this case. Dealing with these practical issues is an important avenue for future work.

\bibliographystyle{IEEEtran}
\bibliography{joohwan}

\appendix
%dummy comment inserted by tex2lyx to ensure that this paragraph is not empty

\section{Proofs}

\subsection{Lemma~\ref{lemma:num_of_candidates}}
\begin{lemma}\label{lemma:num_of_candidates}
At each iteration of the FGC process, we have $$|C(v_t,E^{(1)})|=N-|E^{(1)}|-1=N-t,$$
$$|C(v_t,E^{(2)})|=N-|E^{(2)}|-1=N-\summu{t-1}-1,$$
for $t<N$.
\end{lemma}
\begin{proof}
We first prove by induction that $C(v_t,E^{(1)})$  satisfies the lemma.  Since $E^{(1)}=\emptyset$ in Step~3 of iteration $t=1$, $c_t$ can be any peer but $v_t$. Thus, this lemma holds.

Assume that this lemma holds for $t=l-1$. Since we have added one edge to $E^{(1)}$ in Step~4 of each iteration $t<l$, the number of edges in $E^{(1)}$ in Step~3 of iteration $t=l$ is $l-1$. By Condition \ref{condition:to_be_a_cycle}, the $l-1$ peers where these edges end are not included in $C(v_t,E^{(1)})$.  If none of the $l-1$ peers is $v_t$, peer $v_t$ cannot also be in $C(v_t,E^{(1)})$ because adding $(v_t, v_t)$ incurs a loop. If one of the $l-1$ peers is $v_t$, there must be an acyclic path ending at $v_t$. Adding an edge from $v_t$ to the first peer of the path makes a non-Hamiltonian cycle. Thus, the first peer, which cannot be one of the $l-1$ peers, cannot be in $C(v_t,E^{(1)})$. Overall, there are $(l-1)$ peers that cannot satisfy Condition \ref{condition:to_be_a_cycle}.

The proof for  $C(v_t,E^{(2)})$ is straightforward. At the beginning of iteration $t=l$, the number of edges in $E^{(2)}$ is $\summu{l-1}$. Using the same logic, it is easy to show that $|C(v_t,E^{(2)})|=N-|E^{(2)}|-1$.
\end{proof}

\subsection{Proof of Proposition~\ref{prop:flow_graph_process}}
\label{proof:prop:flow_graph_process}
\tblue{We shows how $H_1$ and $H_2'$ are distributed. Recall that $\cC(V)$ is the set of all possible Hamiltonian cycles that can be made of peer set $V$. From Proposition~\ref{prop:layer_distribution}, $H_1$ is uniformly distributed in  $\cC(V)$. Define $\tilde \cC(V,J)$ to be the set of all possible subgraphs of a graph in $\cC(V)$ with $J$ edges, i.e.,
$$\tilde \cC(V,J)=\{(V, E) ||E|=J, \exists (V, E')\in \cC(V) \text{ s.t. } E\subset E'\}.$$
It is not difficult to show that
$$|\tilde \cC(V,J)|=   {{N}\choose{J}}\frac{(N-1)!}{(N-J-1)!}.$$
(We have first chosen $J$ peers among $N$ peers, and then have chosen the number of ways in which we can draw outgoing edges from them without violating the Hamiltonian cycle constraint.)
Recall that $H_2'$ is obtained by randomly removing each edge with probability $1-q$ from $H_2$, which is also uniformly distributed in $\cC(V)$. Hence, conditioned on the fact that the number  of remaining edges in $H_2'$ is $J$, i.e., $|E(H_2')|=J$, $H_2'$ is uniformly distributed in $\tilde \cC(V,J)$. Thus, conditioned on $|E(H_2')|=J$, the  graph $H^*$ that we obtain by  superposing two independent graphs $H_1$ and $H_2'$ is uniformly distributed in
\begin{align}
&\cC(V) \times\tilde \cC(V,J)\NN\\
&\triangleq \{(V, E', E'')| (V,E')\in \cC(V), (V, E'')\in \tilde \cC(V,J)\}.\NN
\end{align}
From now on, we show how to relate this result to the distribution of $(V,E^{(1)}, E^{(2)})$ in the FGC process. For each graph $G$ that can be $H^*$, we need to show that
\begin{equation}
\prob{H^*=G}= \prob{(V,E^{(1)}, E^{(2)}) =G}.\label{eq:theSame2}
\end{equation}
This is equivalent to showing the following:
\begin{align}
&\sum_{j=0}^N\prob{H^*=G\;\big|\; |E(H_2')|=j}\prob{|E(H_2')|=j}\NN\\
=&\sum_{j=0}^N\prob{(V,E^{(1)}, E^{(2)})=G\;\big|\; \summu{N}=j}\prob{\summu{N}=j}\label{eq:theSame}
\end{align}
Note that  $|E(H_2')|$ is a binomial random variable with parameter $(N,q)$ because we have removed each edge from $H_2$ with probability $(1-q)$. Since $\tau_1,\cdots, \tau_N$ are independent Bernoulli random variables with mean $q$, $\summu{N}$ is also a binomial random variable with the same parameter. Hence, both $\summu{N}$ and $|E(H_2')|$ have the same distribution, i.e.,
$\prob{|E(H_2')|=j}=\prob{\summu{N}=j}$ for all $j$. Hence, if the  graph $(V, E^{(1)}, E^{(2)})$ resulting from the FGC process is uniformly distributed in
$\cC(V) \times\tilde \cC(V,J)$ conditioned on $\summu{N}=J$,
 the equality in (\ref{eq:theSame}) is satisfied, and thus  (\ref{eq:theSame}) is also satisfied. Hence, we now focus on showing that $(V, E^{(1)}, E^{(2)})$ is uniformly distributed in $\cC(V) \times\tilde \cC(V,J)$ conditioned on $\summu{N}=J$.}

\tblue{Since we have constructed $(V, E^{(1)})$ following Condition~\ref{condition:to_be_a_cycle} in the FGC process, it is a Hamiltonian cycle, which belongs to $\cC(V)$. Due to the same reason, $(V, E^{(2)})$ is a subgraph of a Hamiltonian cycle, and thus it should belong to $\tilde \cC(V,J)$ if $\summu{N}=J$. Hence, conditioned on $\summu{N}=J$, the resulting graph $(V, E^{(1)}, E^{(2)})$ belongs to $\cC(V) \times\tilde \cC(V,J)$.}

\tblue{We next show how $(V, E^{(1)}, E^{(2)})$ is distributed in $\cC(V) \times\tilde \cC(V,J)$. Let $\vec c = (c_1, \cdots, c_N)$. For given $\vec \tau$, define $\vec c'=(c'_{t}\; ;\;
\tau_t=1)$. Note that the FGC process randomly chooses $(\vec c, \vec \tau, \vec c')$. We  prove by contradiction that a unique choice of $(\vec c, \vec \tau, \vec c')$ in the FGC process results in a unique  $(V,E^{(1)},E^{(2)})$. Assume to the  contrary that two different choices $\Delta=(\vec  c, \vec \tau, \vec c')$ and $\hat \Delta=(\vec \hat{c}, \vec \hat{\tau}, \vec \hat{c'})$ result in the same graph. Let $t^*$ be the first iteration that both decisions are not the same, i.e.,
$c_t=\hat c_t$, $\tau_t=\hat  \tau_t$, and $c'_{l}=\hat  c'_{l}$ for all $t<t^*$ and $l<t^*$ satisfying $\tau_l=1$, and $c_{t^*}\neq \hat c_{t^*}$ or $\tau_{t^*}\neq \hat \tau_{t^*}$ or $c'_{t^*}\neq \hat c'_{t^*}$ if $\tau_{t^*}= \hat \tau_{t^*}$. Since the FGC process works identically  up to iteration $t-1$ under both choices, $v_t$ at iteration $t$ must be also identical. Since both choices $\Delta$ and $\tilde \Delta$ differ at iteration $t$, the outgoing edges of $v_t$ will be different under both choices. Hence, the resulting graphs under both choices cannot be the same, which is a contradiction. Hence, a unique $(\vec  c, \vec \tau, \vec c')$ results in a unique resulting graph $(V, E^{(1)}, E^{(2)})$, which we have proven to be in $\cC(V)\times \tilde \cC(V,J)$ in the previous paragraph.}

\tblue{We finally show $(V, E^{(1)}, E^{(2)})$ is uniformly distributed in $\cC(V) \times\tilde \cC(V,J)$ conditioned on $\summu{N}=J$.
 From Lemma~\ref{lemma:num_of_candidates}, $c_t$ $(t<N)$ is chosen uniformly at random among $N-t$ candidates. Hence, $\vec c$ is chosen uniformly at random among $(N-1)!$ possible combinations. Similarly, from Lemma~\ref{lemma:num_of_candidates}, $c_t'$ is chosen uniformly at random among $N-\summu{t-1}-1$ candidates, and thus $\vec c'$ is chosen uniformly at random among $(N-1)(N-2)\cdots (N-J)$ combinations if $\summu{N}=J$. Conditioned on $\summu{N}=J$, there are ${{N}\choose{J}}$ candidates for $\vec \tau$. It is easy to see that $\vec \tau$ is uniformly distributed among these candidates. Thus, conditioned on $\summu{N}=J$, the FGC process chooses $(\vec c, \vec \tau, \vec c')$ uniformly at random among ${{N}\choose{J}} \frac{((N-1)!)^2}{(N-J-1)!}$ candidates.
Since a unique $(\vec c, \vec \tau, \vec c')$ results in unique $(V,E^{(1)}, E^{(2)}) \in \cC(V)\times \tilde \cC(V,J)$, $(V,E^{(1)}, E^{(2)})$ is uniformly distributed among ${{N}\choose{J}} \frac{((N-1)!)^2}{(N-J-1)!}$
graphs in $\cC(V)\times \tilde \cC(V,J)$. Since the cardinality of $\cC(V)\times \tilde \cC(V,J)$ is  ${{N}\choose{J}} \frac{((N-1)!)^2}{(N-J-1)!}$, we can say that $(V,E^{(1)}, E^{(2)})$ is uniformly distributed over the entire set $\cC(V)\times \tilde \cC(V,J)$ conditioned on $\summu{N}=J$.}

\tblue{As we mentioned right after (\ref{eq:theSame}), both $(V,E^{(1)}, E^{(2)})$ and $H^*$ have the same distribution.}

\subsection{Proof of Proposition~\ref{prop:N-z(t)_on_k}}

\label{proof:prop:N-z(t)_on_k}

We prove this result by induction. We first show that (\ref{eq:doob_martingale}) is true for $l=t$. Since $\mathcal{G}(t)$ has been defined as the graph drawn up to  iteration $t$, $Z(t)$
and $z(t)$ are deterministic conditioned on $\mathcal{G}(t)$. Thus, we can remove the expectation
from $N-z(t)$, which corresponds to (\ref{eq:doob_martingale}) for
$l=t$.

We assume that (\ref{eq:doob_martingale}) is true for $l=l'+1\leq t$,
i.e.,
%\begin{align}
% & E\big[N-z(t)\big|E^{(1)}(l),E^{(2)}(l)\big]\NN\\
%= & \frac{N-t-1}{N-(k+1)-1}\left(1-\frac{(t-(k+1))q}{N-|E^{(2)}(k+1)|-1}\right)\times\NN\\
%&(N-z(k+1)).\label{eq:doob_extension1}
%\end{align}
\begin{align}
 & \frac{E\big[N-z(t)\big|\cG(l'+1),\tred{\mu} \tblue{\vec \tau}\big]}{(N-t-1)(N-\summu{t}-1)}\NN\\
=&\frac{N-z(l'+1)}{(N-(l'+1)-1)(N-\summu{l'+1}-1)}.\NN
%\label{eq:doob_extension1}
\end{align}
If we rewrite $z(l'+1)$ using (\ref{eq:update_of_zt}) and take $E[\;\cdot\;|\; \cG(l'),\tred{\mu} \tblue{\vec \tau}]$, we have
\begin{align}
 & \frac{E\big[N-z(t)\big|\cG(l'),\tred{\mu} \tblue{\vec \tau}\big]}{(N-t-1)(N-\summu{t}-1)}\NN\\
=&\frac{N-z(l')-E[X_1+ X_2| \cG(l'),\tred{\mu} \tblue{\vec \tau} ]}{(N-(l'+1)-1)(N-\summu{l'+1}-1)}.\label{eq:doob_condition_t-1}
\end{align}
where $X_1=\indicate{c_{l'+1}\notin Z(l')}$ and $X_2=\tred{\indicate{\mu(v_{l'+1})=k}}  \tblue{\tau_{l'+1}}\indicate{c_{l'+1}'\notin Z(l')\cup\{c_{l'+1}\}}$.

In the FGC process, $c_{l'+1}$ is chosen from
$C(v_{l'+1}, E^{(1)}(l'))$ with a uniform distribution. Since
$|C(v_{l'+1}, E^{(1)}(l'))|= N-l'-1$ by Lemma~\ref{lemma:num_of_candidates}, we have
\begin{align}
&E[X_1| \cG(l'),\tred{\mu} \tblue{\vec \tau}]\NN\\
=&\frac{|V\setminus Z(l')|}
{|C(v_{l'+1}, E^{(1)}(l'))|}=\frac{N-z(l')}{N-l'-1}.\label{eq:mean_x1}
\end{align}
Similarly, \tred{if $\mu(v_{l'+1})=k$,} $c'_{l'+1}$ is chosen from
$C(v_{l'+1}, E^{(2)}(l'))$ with uniform distribution. Since $|C(v_{l'+1}, E^{(2)}(l'))|=N-\summu{l'}-1$ by Lemma~\ref{lemma:num_of_candidates}, we have
\begin{align}
&E[X_2| \cG(l'),\tred{\mu} \tblue{\vec \tau} ,X_1]\NN\\
=&\frac{\tblue{\tau_{l'+1}} \tred{\indicate{\mu( v_{l'+1})=k}} |V\setminus (Z(l')\cup\{c_{l'+1}\})|}
{|C(v_{l'+1}, E^{(2)}(l'))|}\NN\\
=&\frac{\tblue{\tau_{l'+1}} \tred{\indicate{\mu( v_{l'+1})=k}} (N-z(l')-X_1)}{(N-\summu{l'}-1)}.\NN
\end{align}
Taking $E[ \;\cdot\; | \cG(l'),\tred{\mu}\tblue{\vec \tau} ]$ to both sides and applying (\ref{eq:mean_x1}), we have
\begin{align}
&E[X_2| \cG(l'),\tred{\mu}\tblue{\vec \tau} ]\NN\\
=&\frac{ \tblue{\tau_{l'+1}}  \tred{\indicate{\mu(v_{l'+1})=k}}(N-z(l'))}{N-\summu{l'}-1}\frac{N-l'-2}{N-l'-1}.\label{eq:mean_x2}
\end{align}
Applying (\ref{eq:mean_x1}) and (\ref{eq:mean_x2}) to the numerator of (\ref{eq:doob_condition_t-1}), we have
\begin{align}
&N-z(l')-E[X_1+ X_2| \cG(l'),\tred{\mu} \tblue{\vec \tau} ]\NN\\
=&(N-z(l'))\Big(1-\frac{1}{N-l'-1}\NN\\
&\;\;\;\;\;\;-\frac{\tblue{\tau_{l'+1}} \tred{\indicate{\mu( v_{l'+1})=k}}  }{N-\summu{l'}-1}
\frac{N-l'-2}{N-l'-1}\Big)\NN\\
=&(N-z(l'))\frac{N-l'-2}{N-l'-1}\frac{N-\summu{l'+1}-1}{N-\summu{l'}-1}.\NN
\end{align}
Adding the above to (\ref{eq:doob_condition_t-1}), we have (\ref{eq:doob_martingale}) for $l=l'$.
By induction, (\ref{eq:doob_martingale})
is true for $1\leq l\leq t$.

\subsection{Proof of Proposition~\ref{prop:expansion_of_flowgraph}}

\label{proof:prop:expansion_of_flowgraph}

We first define the Doob martingale $\{B_{l}\}_{0\leq l\leq t}$ as
$B_{l}\triangleq E[N-z(t)\;|\; \cG(l)]$. From (\ref{eq:update_of_zt}),
we find an upper bound on $|B_{l}-B_{l-1}|$ for each
$l$ and then find the probability that the martingale concentrates
around its mean using the Azuma-Hoeffding bound.

The upper bound $\psi_{l}$ can be found from the following lemma.
\begin{lemma}\label{lemma:lipshitz}
Let $\tilde B_l = E[N-z(t)| \cG(l), \tred{\mu}\tblue{\vec \tau}]$. For $1\leq l\leq t$,
$$|\tilde B_{l}-\tilde B_{l-1}|\leq 2.$$
 \end{lemma} The proof is provided in Appendix~\ref{proof:lemma:lipshitz}. Since the upper-bound is independent of $\tblue{\vec \tau}\tred{\mu}$, the same bound also holds for $B_l$, i.e., $|B_l-B_{l-1}|\leq 2$.
Using the Azuma-Hoeffding bound, we have
\begin{equation}
\prob{B_{t}-B_{0}>\alpha}<\exp\Big(-\frac{\alpha^{2}}{2\sum_{j=1}^{t}2^{2}}\Big).\label{eq:azuma_hoeffding}
\end{equation}
From Proposition~\ref{prop:N-z(t)_on_k}, $B_{t}=N-z(t)$ and $B_{0}=(N-t-1)(1-\frac{tq}{(N-1)})$.
Taking $\alpha=1+t(q(1-\frac{t}{N-1})-\psi)$, (\ref{eq:azuma_hoeffding})
can be expressed as follows: for $t\leq N/2$,
\begin{align}
&\prob{z(t)<(1+\psi)t}  <\exp\Big(-\frac{\alpha^{2}}{8t}\Big)\NN\\
 & <\exp\Big(-\frac{[(\frac{q}{2}-\psi)t]^{2}}{8t}\Big)\label{eq:azuma2}\\
 & =\exp\Big(-\frac{(\frac{q}{2}-\psi)^{2}t}{8}\Big)=\exp(-\sigma t).\NN
\end{align}
In $(\ref{eq:azuma2})$, we have used $\alpha<(q/2-\psi)t$ for $t\leq N/2$.
Thus, we have proven Proposition~\ref{prop:expansion_of_flowgraph}.

\subsection{Proof of Lemma~\ref{lemma:lipshitz}}
\label{proof:lemma:lipshitz}
From Proposition~\ref{prop:N-z(t)_on_k}, we have
\begin{align}
 & \tilde B_{l}-\tilde B_{l-1}\NN\\
= & f(l)(N-z(l))-f(l-1)(N-z(l-1)),\NN
\end{align}
where $$f(l)=\frac{N-t-1}{N-l-1}\frac{N-\summu{t}-1}{N-\summu{l}-1}.$$
Since $l\leq t$, we have $f(l-1)\leq f(l)\leq 1$. Thus,
\begin{align}
&\tilde B_{l}-\tilde B_{l-1}\NN\\
\geq& f(l)(N-z(l)-(N-z(l-1)))\NN\\
\geq& f(l)(-2)>-2.\NN
\end{align}
Since $z(l-1)\leq z(l)$, we have
$$\tilde B_{l}-\tilde B_{l-1}\leq (N-z(l-1))(f(l)-f(l-1)).$$
Taking $x_l \triangleq (N-l-1)$ and $y_l\triangleq (N-\summu{l}-1)$, we can simplify $(f(l)-f(l-1))$ in the above as
\begin{align}
&(f(l)-f(l-1))\NN\\
%\leq & (N-t-1)(N-\summu{t}-1)\NN\\
%&\;\;\;\;\times\Big[\frac{1}{(N-l-1)(N-\summu{l}-1)}\NN\\
%&\;\;\;\;\;\;\;\;\;\;\;\;\;-\frac{1}{(N-l)(N-\summu{l-1}-1)}\Big].\NN\\
\leq & x_t y_t \Big[\frac{1}{x_l y_l}-\frac{1}{(x_l+1)(y_l+\tblue{\tau_l} \tred{\indicate{\mu(v_l)=k}})}\Big]\NN\\
\leq & x_t y_t \Big[\frac{1}{x_l y_l}-\frac{1}{(x_l+1)(y_l+1)}\Big]\NN\\
\leq & x_t y_t \Big[\frac{(x_l+1)(y_l+1)-x_l y_l}{x_l y_l(x_l+1)(y_l+1)}\Big]
\leq  \frac{x_t y_t (x_l+1)+x_t y_t  y_l}{x_l y_l(x_l+1)(y_l+1)}\NN\\
\leq& \frac{1}{y_l+1}+\frac{1}{x_l+1}\leq \frac{2}{x_l+1}.\NN
\end{align}
For the last two inequalities,  we have used $x_t\leq x_l$, $y_t\leq y_l$, and $x_l\leq y_l$.
%The terms in the bracket can be simplified as
%\begin{align}
%&\frac{(N-l)(N-\summu{l-1}-1)-(N-l-1)(N-\summu{l}-1)}
%{(N-l-1)(N-\summu{l}-1)(N-l)(N-\summu{l-1}-1)}\NN\\
%=&\frac{(N-\summu{l}-1)+\indicate{\tred{\mu}(v_l)=k}(N-l)}
%{(N-l-1)(N-\summu{l}-1)(N-l)(N-\summu{l-1}-1)}\NN\\
%=&\frac{1}{(N-l-1)(N-\summu{l}-1)(N-l)   }\NN\\
%&\;\;\;\times\frac{(N-\summu{l}-1)+\indicate{\tred{\mu}(v_l)}(N-l)}
%{(N-\summu{l-1}-1)}\NN\\
%<&\frac{1+\indicate{\tred{\mu}(v_l)}}{(N-l-1)(N-\summu{l}-1)(N-l)}.\NN
%\end{align}
Since $N-z(l-1)\leq N-l=x_l+1$,  we have $\tilde B_{l}-\tilde B_{l-1}\leq 2$.
Overall, we can find the Lipschitz difference $|\tilde B_{l}-\tilde B_{l-1}|\leq 2.$

\end{document}